\documentclass[12pt]{article}

\usepackage{amsmath, amsthm, amssymb}
\usepackage{graphicx}
\usepackage{booktabs}
\usepackage{multirow}
\usepackage{dcolumn}
\usepackage{mathrsfs}    
\usepackage{float}
\usepackage{subcaption}
\usepackage{tikz}
\usetikzlibrary{arrows.meta,positioning,shapes.geometric}
\usepackage{enumerate}
\usepackage{titlesec}
\usepackage{xcolor}
\usepackage{optidef}
\usepackage{makecell}
\usepackage{latexsym}
\usepackage{bbm}
\usepackage{hyperref}
\usepackage[numbers,sort&compress]{natbib}

\def \be {\begin{equation}}
\def \ee {\end{equation}}

\newcommand{\Tr}{\mathrm{Tr}}
\newcommand{\ket}[1]{|#1\rangle}
\newcommand{\bra}[1]{\langle#1|}

\newcommand{\rev}[1]{{#1}}

\def \sofc2{{\cal S}({\mathbb C}^2)}

\def\>{\rangle}
\def\<{\langle}

\newtheorem{theorem}{Theorem}
\newtheorem{lemma}[theorem]{Lemma}

\newtheorem{corollary}[theorem]{Corollary}

\let\originalref\ref
\renewcommand{\ref}[1]{\originalref{#1}}
\newenvironment{proofof}[1]{\vspace*{5mm} \par \noindent{\it Proof of #1:\hspace{2mm}}}{\endproof
\hfill$\Box$ \vspace*{3mm}}

\renewenvironment{abstract}{%
  \begin{center}%
    \bfseries\large Abstract%
  \end{center}%
  \quotation%
}{%
  \endquotation%
}

\title{Characterization-free identification of the environment between two quantum players}

\author{%
\parbox{0.95\textwidth}{\centering
Masahito Hayashi$^{1,2,3\ast\dagger}$, Longyang Cao$^{4,5\dagger}$, Baichu Yu$^{4,6\ast}$, Yuan-Yuan Zhao$^{4\ast}$\\[0.5ex]
\small $^{1}$School of Data Science, The Chinese University of Hong Kong, Shenzhen, Longgang District, Shenzhen 518172, China\\
\small $^{2}$International Quantum Academy, Futian District, Shenzhen 518048, China\\
\small $^{3}$Graduate School of Mathematics, Nagoya University, Nagoya 464-8602, Japan\\
\small $^{4}$Quantum Science Center of Guangdong-Hong Kong-Macao Greater Bay Area, Shenzhen 518045, China\\
\small $^{5}$Shenzhen University, Shenzhen 518060, China\\
\small $^{6}$Shenzhen Institute for Quantum Science and Engineering, Southern University of Science and Technology, Nanshan District, Shenzhen 518055, China\\[0.5ex]
\small $^{\ast}$Corresponding authors: hmasahito@cuhk.edu.cn; yubc@sustech.edu.cn; zhaoyuanyuan1@quantumsc.cn\\
\small $^{\dagger}$These authors contributed equally to this work.}
}

\date{}

\begin{document}

\maketitle

\begin{abstract}
Identifying the causal structure of quantum channels is essential for verifying quantum networks and certifying quantum resources. We introduce a characterization-free protocol enabling two isolated players, Alice and Bob, to identify the definite-order strategy adopted by an unknown environment mediating their channels. Without assuming knowledge of their devices or the environment, the players infer the causal order solely from input-output statistics by testing Markovian conditions that we prove are necessary and sufficient for each strategy class. Remarkably, we prove that, under an explicit generic-sampling condition, a randomly selected binary measure-and-prepare setting retains exact-distribution identifiability with probability one. In the optical experiment, we use a reduced-randomness construction in which several preparation states are kept fixed. Nevertheless, the Markov-condition-based procedure yields the expected causal-order and memory-presence classification for every tested process realization and setting. This observation suggests that the randomization assumptions of the general theorem may be relaxed. Our results provide an operational framework for causal inference in quantum networks.
\end{abstract}

\section*{Introduction}

The study of quantum channels is central to quantum information science, providing fundamental insights into how information is processed, transmitted, and manipulated in quantum systems. When a quantum protocol involves the application of multiple channels, there exist different strategies that govern their application order. 
Generally, the strategies for applying multiple channels can be classified into definite-order and indefinite-order ones \cite{OCB}. 
In definite-order strategies, the temporal sequence of channel applications is well-defined. In contrast, indefinite-order strategies do not correspond to a well-defined temporal sequence and can exhibit superposition of orders \cite{OCB,goswami2018indefinite}.
While recent work has demonstrated the theoretical interest and operational advantages of indefinite-order strategies
\cite{ebler2018enhanced,zhao2020quantum,kristjansson2020resource,chapeau2021noisy,yin2023experimental,felce2020quantum,xia2026scaling}, definite-order strategies remain the foundation of most quantum circuits, protocols, and networks.   

The ability to control and verify the order of quantum channels is of foundational and operational importance. For instance, in a Bell test entangled across multiple degrees of freedom \cite{miklin2022exponentially}, 
non-demolishing measurements must precede demolishing ones to preserve correlations. More generally, the specific order of operations can directly influence protocol performance: sequential channel arrangements can outperform parallel ones in channel discrimination tasks \cite{bavaresco2022unitary}, 
and optimized causal sequences can enhance cloning fidelity and accelerate quantum-algorithm learning \cite{chiribella2008quantum}. 

Although channel application strategies are often well controlled and assumed known, the need to identify or detect them is growing rapidly. 
For instance, verifying that quantum gates are executed in their intended order—despite noise and errors—is essential for circuit reliability.
Furthermore, in communication protocols with an adversary, the type of strategy constrains the possible attacks.
For instance, if the adversary has a quantum memory and entangles an ancilla with the systems of the communicating parties, then coherent (entangling) attacks become possible~\cite{RennerThesis,FuchsGisinGriffithsNiuPeres1997}.
In fact, in secure key distillation \cite{DevetakWinter2005,renner2008security} or secure communication \cite{scarani2009security}
protocols based on correlated states between communicating parties, the achievable secret key rate and secure communication rate depend on whether these correlations are entangled with the adversary's system. 
In all these cases, the strategies need to be treated as unknown environments, and knowledge about them can only be inferred from experiment.

A previous work \cite{BMQ} provided a mathematical framework for characterizing these strategies in the simplest two-player scenario, but did not address their operational identification. Related operational approaches have used measure-and-prepare interventions to characterize temporal memory in multi-time processes of a single quantum system: Pollock et al.~\cite{PollockOperationalMarkov} introduced causal breaks to test quantum Markovianity, while Taranto et al.~\cite{TarantoQuantumMarkovOrder} formulated instrument-specific quantum Markov order. The time-ordering problem for sequential channels without memory was also studied in Ref.~\cite{Liu2} using pseudo-density matrices and characterized measurement operators; see also Refs.~\cite{FJV,liu2025quantum,song2024causal}.

In this work, we propose a characterization-free protocol that enables two isolated quantum players to identify the definite-order strategies implemented by their environment. Without characterization of their experimental devices, the players determine the underlying strategy using only input--output statistics. We establish a one-to-one correspondence between strategy classes, described via process matrices \cite{OCB,Liu}, and Markov conditions derived from the corresponding input--output statistics. Specifically, these conditions are expressed in terms of the observed preparation and measurement outcomes introduced below. We prove that they become necessary and sufficient for strategy identification under a weak tomographic-completeness assumption. Compared with quantum process tomography \cite{chuang1997prescription,rubino2017experimental}, our protocol avoids requiring explicit calibrated operator-level descriptions of the players' devices in the regimes covered by (S1) and (S2).
Remarkably, we further show that, under a much weaker generic-sampling assumption requiring no characterization of the measurement and preparation operators, a randomly selected binary measure-and-prepare setting retains exact-distribution identifiability with probability one. The sampling condition is a simple sufficient condition for generic identifiability and is not claimed to use the minimum amount of randomization. Our optical implementation further explores a reduced-randomness construction outside the direct scope of this theorem and obtains the expected Markov-condition-based classification in all tested cases.

\section*{RESULTS}
\subsection*{Our Protocol}
\begin{figure}[htbp]
    \centering
    \includegraphics[width=\columnwidth]{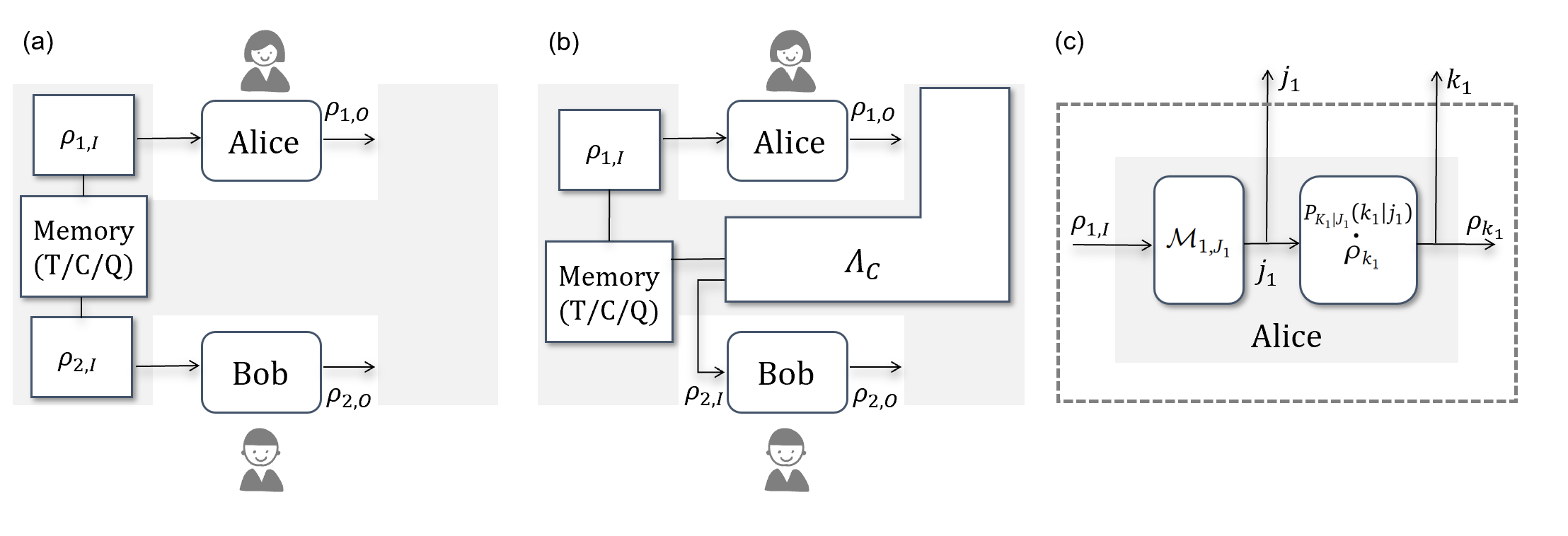}
    \caption{Figs.~(a) and~(b) specify the structure of Charlie's strategy classes within the general class $\mathcal{S}_{G}$, which can be represented as a quantum comb (the gray part). \cite{chiribella2008quantum}. Fig.~(a) specifies the three parallel strategies. When the type of memory is trivial (T)/ classical (C)/ quantum (Q), the strategy is individual ($\mathcal{S}_{I}$)/classical ($\mathcal{S}_{C}$)/quantum ($\mathcal{S}_{Q}$). The unidentified parts in the comb (area in the dashed boundary) do not influence the protocol. Fig.~(b) specifies the three sequential strategies ($\mathcal{S}_{N,1\rightarrow 2}$)/$\mathcal{S}_{C,1\rightarrow 2}$)/$\mathcal{S}_{Q,1\rightarrow 2}$), which correspond to T/ C/ Q type of memory respectively. Fig.~(c) illustrates Alice's MP channel (Bob's is analogous). In our protocol, the parts inside the gray dashed box do not require characterization.}
    \label{fig1}
\end{figure}

We consider two isolated quantum players, Alice and Bob, each controlling a channel with input-output systems $\mathcal{H}_{I,1}/\mathcal{H}_{O,1}$ and $\mathcal{H}_{I,2}/\mathcal{H}_{O,2}$, respectively. 
The systems have dimensions $d_{(I,1)}$, $d_{(O,1)}$, $d_{(I,2)}$, $d_{(O,2)}$ respectively.
A third party, Charlie (the environment), mediates the interaction of their channel 
by arranging their input and output states.
In our scenario, Charlie performs six different definite-order strategies of two different types: parallel and sequential. 
In a parallel strategy, Charlie gives predetermined
input states to Alice and Bob 
which are independent of the output states of their channels, 
while in a sequential strategy the input state of Alice (Bob) can be dependent on Bob's (Alice's) output state. 
Further, Charlie can also apply 
a classical or quantum memory 
to generate some correlations between 
the input states of Alice and Bob.

Alice and Bob have no prior knowledge to the environment, 
nor can they directly observe the operation of Charlie or communicate to each other during the experiment. 
They are restricted to performing measurements 
respectively on the input systems they receive, 
and prepare an output system respectively to send back to the environment. 
Through this two-step measure-and-prepare process, 
Alice and Bob obtain a set of input-output statistics. 
Alice and Bob need to identify Charlie's strategy from such experimental statistics. 

Now we introduce the details and the mathematical expressions of Charlie's strategies, Alice's and Bob's channels, and the experimental correlations. 
In our protocol, Charlie's strategies are classified into six different classes,
each represented by a process matrix 
$W_{\mathcal{S}}$ \cite{OCB}, 
which is a positive semi-definite matrix on 
${\cal H}_{I,1}\otimes {\cal H}_{O,1}\otimes {\cal H}_{I,2}\otimes {\cal H}_{O,2}$.
In general, all physical process matrices satisfy no-signaling-in-time conditions \cite{BMQ}:
\begin{align}
    &(\Tr_{(I,1),(O,1),(O,2)}W_{\mathcal{S}})\otimes\rho_{\text{mix},(O,2)} = \Tr_{(I,1),(O,1)}W_{\mathcal{S}}, \label{G1}\\
    &(\Tr_{(O,1),(I,2),(O,2)}W_{\mathcal{S}})\otimes\rho_{\text{mix},(O,1)} = \Tr_{(I,2),(O,2)}W_{\mathcal{S}},\label{G2}\\
    &W_{\mathcal{S}} = (\Tr_{(O,1)}W_{\mathcal{S}})\otimes\rho_{\text{mix},(O,1)} + (\Tr_{(O,2)}W_{\mathcal{S}})\otimes\rho_{\text{mix},(O,2)} \nonumber\\
    &\quad - (\Tr_{(O,1),(O,2)}W_{\mathcal{S}})\otimes\rho_{\text{mix},(O,1),(O,2)},\label{G3}
\end{align}
where $\rho_{\text{mix}}$ denotes the maximally mixed state. Hereinafter, we call the class of strategies  satisfying Eqs.~\eqref{G1}--\eqref{G3} general strategy class $\mathcal{S}_{G}$. 
The six strategy classes that Charlie may adopt in our scenario are all subclasses of $\mathcal{S}_{G}$.
We first list the three parallel strategies. 

\paragraph*{(i). Individual strategy $\mathcal{S}_{I}$:} Charlie sends independent states $\rho_1$, $\rho_2$ to Alice and Bob as the input states $\rho_{1,I}$ and $\rho_{2,I}$. The process matrix factorizes as:
\begin{equation}\label{SP1}
    W_{\mathcal{S}_{I}} = \rho_1 \otimes \rho_2 \otimes I_{(O,1),(O,2)},
\end{equation}
where $I_{(O,1)(O,2)}$ is the identity matrix on $ {\cal H}_{O,1}\otimes {\cal H}_{O,2}$. 

\paragraph*{(ii). Classical parallel strategy $\mathcal{S}_{C}$:} Charlie uses classical memory $X$ to send correlated states $\rho_{1,X}$, $\rho_{2,X}$. The reduced process matrix is separable:
\begin{equation}\label{SP2}
    W_{\mathcal{S}_{C}} = \sum_X p_X \rho_{1,X} \otimes \rho_{2,X} \otimes I_{(O,1),(O,2)}.
\end{equation}

\paragraph*{(iii). Quantum parallel strategy $\mathcal{S}_{Q}$:} Charlie prepares an entangled state $\rho_{12}$ on $\mathcal{H}_{I,1}\otimes\mathcal{H}_{I,2}$. The process matrix satisfies:
\begin{equation}\label{SP3}
W_{\mathcal{S}_{Q}} =\rho_{12}\otimes I_{(O,1),(O,2)}.
\end{equation}

Then we list the three sequential strategies. 
We only describe strategies with direction $1\rightarrow2$ (Alice$\rightarrow$Bob); 
the opposite direction is analogous.

\paragraph*{(i). Non-memory sequential $\mathcal{S}_{N,1\rightarrow2}$:} Charlie sends $\rho_1$ to Alice, applies channel $\Lambda_C$ to her output, and sends the result to Bob. The process matrix factorizes:
\begin{equation}\label{SS1}
    W_{\mathcal{S}_{N,1\to2}} = \rho_1 \otimes C[\Lambda_{C}] \otimes I_{(O,2)},
\end{equation}
where $C[.]$ represents the Choi matrix of a channel \cite{choi1975completely}. $C[\Lambda_{C}]$ is on space $\mathcal{H}_{O,1} \otimes \mathcal{H}_{I,2}$.

\paragraph*{(ii). Classical sequential $\mathcal{S}_{C,1\rightarrow2}$:} With classical memory $X$, Charlie sends $\rho_{1,X}$ and applies $\Lambda_{C,X}$. The reduced matrix is separable:
\begin{equation}\label{SS2}
    W_{\mathcal{S}_{C,1\to2}} = \sum_X p_X \rho_{1,X} \otimes C[\Lambda_{C,X}]\otimes I_{(O,2)},
\end{equation}

\paragraph*{(iii).Quantum sequential $\mathcal{S}_{Q,1\rightarrow2}$:} Charlie prepares entangled $\rho_{1,\text{aux}}$, sends the $\rho_{1}$ part to Alice as the input state, and applies quantum channel $\Lambda_C$ on the joint system of $\mathcal{H}_{aux}$ and Alice's output system $\mathcal{H}_{O,1}$. The process matrix satisfies:
\begin{equation}\label{SS3}     W_{\mathcal{S}_{Q,1\to2}}=\frac{1}{d_{O,2}}(\Tr_{(O,2)}W_{\mathcal{S}_{Q,1\to2}})\otimes I_{(O,2)} .
\end{equation}
We illustrate all parallel and sequential strategies in Fig.~\ref{fig1}. 

Charlie's strategy classes form a hierarchy:
\begin{align*}
\begin{array}{ccccc}
\mathcal{S}_{N,1\to2} & \supset &\mathcal{S}_{I} &\subset & \mathcal{S}_{N,2\to1} \\
\cap & &\cap & & \cap \\ 
\mathcal{S}_{C,1\to2} & \supset &\mathcal{S}_{C} &\subset & \mathcal{S}_{C,2\to1}\\ 
\cap & &\cap & & \cap \\
\mathcal{S}_{Q,1\to2} & \supset &\mathcal{S}_{Q} &\subset & \mathcal{S}_{Q,2\to1}
\end{array}.
\end{align*}\label{BSM}

Now we specify the two-step operation of Alice and Bob.
We denote such operation as measure-and-prepare (MP) channel. 
An MP channel contains two types of random variables corresponding to the measurement outcomes and state preparation.
Hereinafter, we let 
$J_{1}$ and $J_{2}$ describe the random variable corresponding to the measurement outcomes of Alice and Bob, and 
$K_{1}$, $K_{2}$ describe the (classical label of the) states they output, respectively.
And we let $\mathcal{B}(\mathcal{H})$ 
denote the space of bounded Hermitian operators 
on a specific Hilbert space $\mathcal{H}$.
The two-step operation of MP channel is specified as follows.
\begin{description}
\item[(i) \emph{Measurement step}] When Alice and Bob respectively receive the unknown input states $\rho_{1}$ and $\rho_{2}$, they make measurements $\mathcal{M}_{1,J_1}:=\{M_{1,j_1}\}$ and $\mathcal{M}_{2,J_2}:=\{M_{2,j_2}\}$ on them, where $M_{1,j_1}\in \mathcal{B}(\mathcal{H}_{I,1})$ and $M_{2,j_2}\in \mathcal{B}(\mathcal{H}_{I,2})$ are POVM operators giving outcome $j_1$, $j_2$.
\item[(ii) \emph{Preparation step}] 
Alice (1st player) and Bob 
(2nd player) prepare quantum states $\rho_{k_1}$ and $\rho_{k_2}$ according to the conditional distributions $P_{K_1|J_1}(k_1|j_1)$ and $P_{K_2|J_2}(k_2|j_2)$, and send them to Charlie. 
\end{description}
Note that throughout this paper, 
Alice and Bob will always let
their conditional distributions $P_{K_1|J_1}$ and $P_{K_2|J_2}$  satisfy
\begin{gather}
		P_{K_i|J_i}(k_i|j_i)>0, \label{positivity}\\
		 P_{K_i|J_i=m}\neq P_{K_i|J_i=n}, \ \forall n\neq m,\  (i=1,2). \label{nontrivial} 
\end{gather}
Condition \eqref{positivity} guarantees the positivity of all conditional distributions. Condition \eqref{nontrivial} 
prevents trivial MP channel 
that could not distinguish sequential strategies.
Generally, any quantum channel with measurements can be mathematically represented by a quantum instrument.
Let $\{ \Gamma_{j_1,k_1} \}_{j_1,k_1}$ denote the quantum instrument corresponding to the MP channel,
where every element $ \Gamma_{j_1,k_1}$ is a completely positive, trace-nonincreasing map,
and $\sum_{j_1,k_1}\Gamma_{j_1,k_1}$
is a completely positive, trace-preserving map.
The element $\Gamma_1=\{\Gamma_{1,(j_1,k_1)}\}_{(j_1,k_1)}$
and $\Gamma_2=\{\Gamma_{2,(j_2,k_2)}\}_{(j_2,k_2)}$ 
of Alice and Bob  
are written as
\begin{align}\label{instruments}
\Gamma_{1,(j_1,k_1)}(\rho):= &P_{K_1|J_1}(k_1|j_1)
(\Tr \rho M_{1,j_1} )\rho_{k_1} \\
\Gamma_{2,(j_2,k_2)}(\rho):= &
P_{K_2|J_2}(k_2|j_2)
(\Tr \rho M_{2,j_2} )\rho_{k_2}.
\end{align}

With a process matrix $W_{\mathcal{S}}$ and a quantum instrument,
the experimental correlations $P_{J_1,J_2,K_1,K_2}(j_1,j_2,k_1,k_2)$ can be obtained by the generalized Born's rule as \cite{OCB}

\begin{equation}
    P_{J_1,J_2,K_1,K_2}(j_1,j_2,k_1,k_2) = \Tr[W_{\mathcal{S}} (C[\Gamma_{1,(j_1,k_1)}] \otimes C[\Gamma_{2,(j_2,k_2)}])].
    \label{eq:GBR}
\end{equation}
We illustrate the structure of Alice's MP channel in Fig.~\ref{fig1} (Bob's is analogous). In our protocol, characterization of the experimental devices is unnecessary, i.e., the result of identification relies only on the correlations $P_{J_1,J_2,K_1,K_2}$, 
which is independent of the knowledge of the channels. 

\subsection*{Strategy Identification by Statistics}

After obtaining the correlations $P_{J_1,J_2,K_1,K_2}$,
we identify the causal order and the presence or absence of memory by
testing the Markov conditions among the random variables
$J_1,J_2,K_1,K_2$. In fact, the classical memory and quantum memory can also be distinguished by detecting the nonlocality using the correlations, but we do not consider such a problem in this paper.
We use the standard graphical notation for conditional independence.
For random variables $X,Y,Z$, the chain $X-Y-Z$ means that $X$ and $Z$
are conditionally independent given $Y$, i.e.,
$P_{X,Z|Y}=P_{X|Y}P_{Z|Y}$ (equivalently $X \perp Z \mid Y$).
We also write $A \perp B$ for (unconditional) independence.
We now introduce the Markov-condition-based identification.
First we consider the following 
condition on the operations
of Alice and Bob:
\begin{description}
\item[(S1)]
The sets 
$\{M_{1,j_1}\}$, $\{M_{2,j_2}\}$,
$\{\rho_{k_1}\}$, $\{\rho_{k_2}\}$ 
span the spaces ${\cal B}({\cal H}_{(I,1)})$,
${\cal B}({\cal H}_{(I,2)})$,
${\cal B}({\cal H}_{(O,1)})$, and
${\cal B}({\cal H}_{(O,2)})$, respectively.
This condition is called 
tomographically complete.
\end{description}
In this case, 
the sets 
$\{C[\Gamma_{1,(j_1,k_1)}]\}_{(j_1,k_1)}$ and
$\{C[\Gamma_{2,(j_2,k_2)}]\}_{(j_2,k_2)}$ 
span the spaces 
${\cal B}({\cal H}_{(I,1)}\otimes {\cal H}_{(O,1)})$,
and
${\cal B}({\cal H}_{(I,2)}\otimes {\cal H}_{(O,2)})$,
respectively.
Then, we have the following equivalence conditions.
\begin{theorem}\label{TH4}
Assume the condition (S1).
A physical process matrix $W_{\mathcal{S}}$ satisfies
the conditions \eqref{SP1}, \eqref{SP3}, 
\eqref{SS1}, 
and \eqref{SS3}, 
if and only if
the distribution $P_{J_1,J_2,K_1,K_2}$ 
satisfies the Markovian conditions
$(J_1,K_1)\perp(J_2,K_2)$,
$K_1- J_1- J_2-K_2$,
$J_1- K_1- J_2-K_2$, 
and $(J_1, K_1)- J_2-K_2$, respectively.
The same statements hold when exchanging the labels $1$ and $2$
(for direction $2\to 1$).
\end{theorem}
The proof of theorem \ref{TH4} is provided in the Supplemental Materials. 
Theorem \ref{TH4} establishes a one-to-one correspondence between the hierarchies 
\begin{align}
\begin{array}{ccccc}
\mathcal{S}_{N,1\to 2} & \supset &\mathcal{S}_{I}&\subset & \mathcal{S}_{N,2\to 1} \\
\cap &  &\cap & & \cap  \\ 
\mathcal{S}_{Q,1\to 2} & \supset &\mathcal{S}_{Q}&\subset & \mathcal{S}_{Q,2\to 1} 
\end{array}\label{HI1}
\end{align} 
and the Markov Chain structure hierarchies 
\begin{align}
\begin{array}{ccccc}
J_1- K_1- J_2-K_2 & \supset &(K_1,J_1) \perp (J_2,K_2)&\subset & K_1- J_1- K_2-J_2 \\
\cap &  &\cap & & \cap  \\ 
(J_1, K_1)- J_2-K_2 & \supset &K_1-J_1-J_2-K_2 &\subset & K_1- J_1- (J_2,K_2)
\end{array}
.\label{HI4}
\end{align}
Therefore it provides 
a necessary and sufficient condition 
to identify the strategy classes 
in hierarchies \eqref{HI1} 
by identifying
their corresponding Markovian conditions
in \eqref{HI4}. 

The identification of Markovian conditions can be done 
by making hypothesis testing \cite{8231191}.
In our protocol,
we use $\chi^{2}$ tests to 
do the hypothesis testing, 
for which details are put in the Methods section. 
The identification 
starts from the lower level classes to the higher level classes. 
If the Markovian conditions of a lower 
class is accepted,
the strategy belongs to this class, 
otherwise, we go on checking the Markovian conditions
of other same level classes and higher level classes
until a Markovian condition is accepted.
Since the Markovian condition does not depend on the 
device characterization, this protocol 
is characterization-free.
The candidate Markov conditions in Eq.~\eqref{HI4} are tested from
lower to higher levels of the class hierarchy.  If a lower-level condition
is rejected, the procedure continues to the next candidate class.  See
Fig.~\ref{Process} and the Supplemental Materials for further details.

\begin{figure}
    \centering
    \includegraphics[width=\columnwidth]{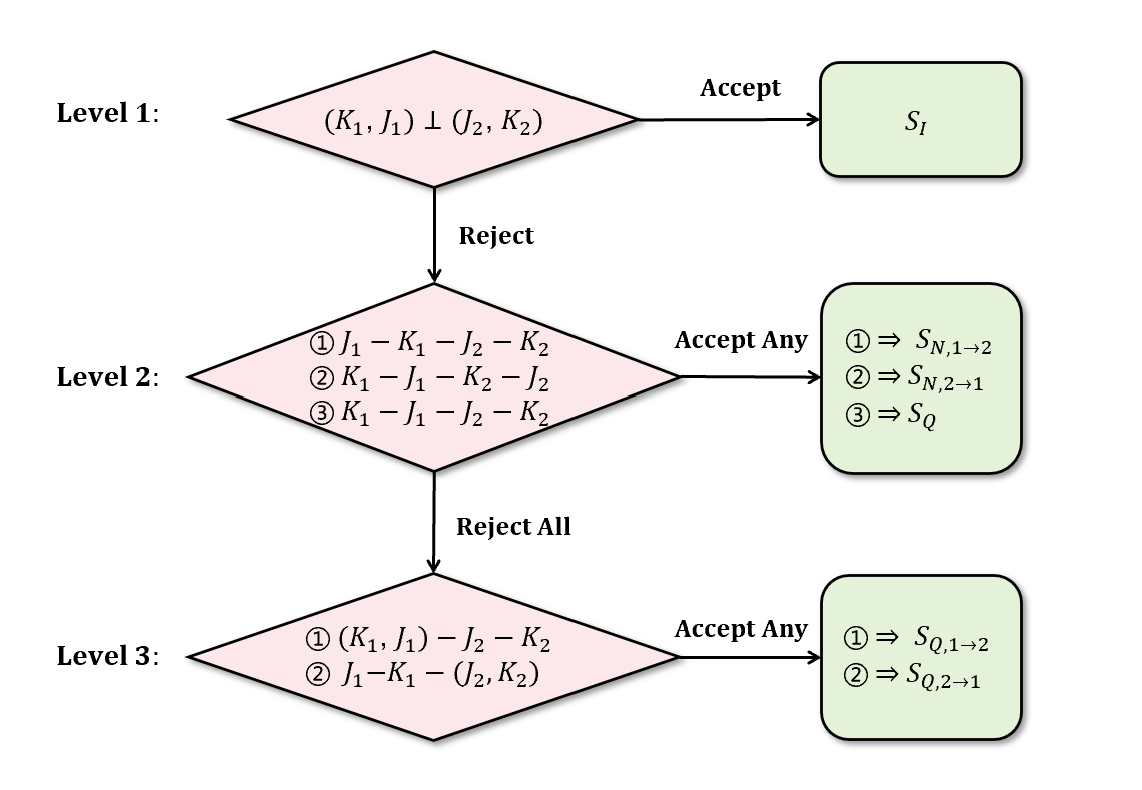}
    \caption{Markov-condition-based classification using the correlations $P_{J_1,J_2,K_1,K_2}$. The Markov conditions in Eq.~\eqref{HI4} are organized into three levels according to the corresponding class hierarchy and are tested from the lowest level upward. If all candidate conditions at one level are rejected, the procedure proceeds to the next level.}
    \label{Process}
\end{figure}

We next consider a fixed process matrix chosen independently of the
randomly sampled MP setting.  We show that, under the following sampling
condition, a single binary random POVM and two randomly sampled preparation
states for each player exclude an incorrect strategy class with probability
one at the level of the exact induced distribution.
\begin{description}
	\item[(S2)]
	The six operators $M_1,M_2,\rho_{1,1},\rho_{1,2},\rho_{2,1},\rho_{2,2}$
	are sampled independently. We collect them into the random tuple
	\begin{equation}
		\boldsymbol{\Xi}:=\left(M_1,M_2,\rho_{1,1},\rho_{1,2},\rho_{2,1},\rho_{2,2}\right).
		\label{eq:random-operator-tuple}
	\end{equation}
	For each component, its sampling law is absolutely continuous with respect
	to a unitarily invariant probability measure on the corresponding set of
	positive operators.  This reference measure assigns probability zero to
	operators proportional to the identity.  For a measurement component the
	operator set is $\{M:0\leq M\leq I\}$, whereas for a preparation component
	it is the set of density operators on the corresponding Hilbert space.
	Alice and Bob use the binary POVMs $\{M_i,I_{I,i}-M_i\}$ and the two
	preparation states $\rho_{i,1},\rho_{i,2}$. Hence $J_i,K_i\in\{1,2\}$.
	The conditional preparation probabilities satisfy
	$P_{K_i|J_i}(k_i|j_i)>0$ for all $j_i,k_i$ and $i=1,2$.
\end{description}

Similarly, the protocol satisfying Condition (S2) is also
characterization-free.
Then we have the following theorem:

\begin{theorem}\label{TH5}
Assume that
(i) Charlie's strategy does not depend on Alice's and Bob's operations, and
(ii) Alice and Bob draw one random binary measure-and-prepare setting
$\boldsymbol{\Xi}$ according to condition~(S2).
Let $\mathcal S_X$ be one of
$\mathcal S_I$, $\mathcal S_Q$, $\mathcal S_{N,1\to2}$,
$\mathcal S_{N,2\to1}$, $\mathcal S_{Q,1\to2}$, and
$\mathcal S_{Q,2\to1}$.  If Charlie's strategy does not belong to
$\mathcal S_X$, then, with probability $1$ over the draw of
$\boldsymbol{\Xi}$, the exact distribution of
$(J_1,K_1,J_2,K_2)$ induced by the selected setting does not satisfy
$Mar(\mathcal S_X)$.
\end{theorem}

The probability in Theorem~\ref{TH5} is taken over the random
selection of the MP setting.  Its conclusion concerns the exact distribution
induced by the selected setting.  The independent randomization of all six
operator components in condition~(S2) is used as a sufficient condition and
is not claimed to be necessary.  The theorem does not provide a finite-sample
power or sample-complexity bound.
Theorem \ref{TH5} provides a
characterization-free method of strategy identification with
a single binary measure-and-prepare setting, which produces
$4\times4=16$ joint instrument outcomes.  For almost every draw satisfying
(S2), these outcomes retain the exact-distribution information required to
exclude an incorrect strategy class.

Next, we show that, when the candidate classes are restricted to $\mathcal S_I$, $\mathcal S_Q$, $\mathcal S_{N,1\to2}$, and $\mathcal S_{Q,1\to2}$, only three components of the measure-and-prepare setting need to be randomized.
\begin{theorem}[Partial-randomization identification]\label{thm:partial-randomization}
Let $W\in\mathcal S_{Q,1\to2}$ be a fixed physical process independent of the setting. Fix normalized $\rho_{1,1},\rho_{2,1},\rho_{2,2}$ and independently sample $M_1,M_2,\rho_{1,2}$ under their respective componentwise assumptions in (S2).  
Let $\Xi' := (M_1, M_2, \rho_{1,2})$ denote the random partial setting. Then,
with probability 1 over the draw of $\Xi'$:
\begin{enumerate}
\item If $W\in\mathcal S_I$, the induced distribution satisfies $Mar(\mathcal S_I)$.
\item If $W\in\mathcal S_Q\setminus\mathcal S_I$, the induced distribution satisfies $Mar(\mathcal S_Q)$ but not $Mar(\mathcal S_I)$.
\item If $W\in\mathcal S_{N,1\to2}\setminus\mathcal S_I$, the induced distribution satisfies $Mar(\mathcal S_{N,1\to2})$ but not $Mar(\mathcal S_Q)$.
\item If $W\in\mathcal S_{Q,1\to2}\setminus(\mathcal S_Q\cup\mathcal S_{N,1\to2})$, the induced distribution satisfies $Mar(\mathcal S_{Q,1\to2})$ but neither $Mar(\mathcal S_Q)$ nor $Mar(\mathcal S_{N,1\to2})$.
\end{enumerate}
This does not assert equal finite-sample power or cover predetermined settings.
\end{theorem}

\rev{To obtain a single reduced-randomization scheme that applies without prior knowledge of the direction, we fix one preparation state for each player and independently randomize the other preparation state, together with the two measurement operators. Theorem \ref{thm:partial-randomization} and its counterpart obtained by exchanging the roles of the two players then apply to the $1\to2$ and $2\to1$ branches, respectively.}
\begin{corollary}\label{TH6}
	 Fix normalized states $\rho_{1,1},\rho_{2,1}$ and independently sample $M_1,M_2,\rho_{1,2}, \rho_{2,2}$ under their respective componentwise assumptions in (S2).  
		Let $\Xi^{*} := (M_1, M_2, \rho_{1,2}, \rho_{2,2})$ denote the random partial setting. Then,  
		If Charlie's strategy does not belong to
		$\mathcal S_X$ defined in Theorem \ref{TH5}, with probability $1$ over the draw of
		$\boldsymbol{\Xi^{*}}$, the distribution 
		$(J_1,K_1,J_2,K_2)$ induced by the selected setting does not satisfy
		$Mar(\mathcal S_X)$.
\end{corollary}

\begin{table*}[htbp]
\caption{Comparison between existing methods.  Here memory identification means identifying its presence or absence; classical and quantum memory are not distinguished by the present protocol.}
\label{comparison}
\begin{center}\resizebox{\textwidth}{!}{%
\begin{tabular}{lccccc}\hline
\textbf{Method} & \makecell{\textbf{Quantum Process}\\\textbf{Tomography}} & \makecell{\textbf{Pseudo-Density}\\\textbf{Matrix \cite{FJV}}} & \textbf{Our Protocol (S1)} & \textbf{Our Protocol (S2)}\\\hline
\textbf{Device Knowledge} & Measurements \& States & Measurements & Characterization-free & Characterization-free\\\hline
\textbf{Cost (Settings)} & $d_{(I,1)}^2d_{(I,2)}^2d_{(O,1)}^2d_{(O,2)}^2$ & $d_{(O,1)}^2d_{(O,2)}^2$ & $d_{(I,1)}^2d_{(I,2)}^2d_{(O,1)}^2d_{(O,2)}^2$ & \makecell{One binary MP setting\\16 joint outcomes}\\\hline
\textbf{Order Identification} & Yes & Yes & Exact under (S1) & \makecell{Probability one under (S2)\\for the exact distribution}\\\hline
\textbf{Presence or absence of memory} & Yes & Not applicable & Exact under (S1) & \makecell{Probability one under (S2)\\for the exact distribution}\\\hline
\textbf{Classical vs. quantum memory} & In principle & Not applicable & Not distinguished & Not distinguished\\\hline
\end{tabular}%
}
\end{center}
\end{table*}

\subsection*{Experimental Markov-condition-based classification with a single photonic setup}
To demonstrate the Markov-condition-based classification on an optical
platform, we implement representative realizations of the individual,
parallel-memory, memoryless sequential, and memory-allowing sequential
classes.  The implemented realizations include
$\mathcal S_I$, $\mathcal S_C$, $\mathcal S_Q$,
$\mathcal S_{N,1\to2}$, $\mathcal S_{C,1\to2}$, and
$\mathcal S_{Q,1\to2}$.  In the present classification,
$\mathcal S_C$ and $\mathcal S_Q$ are treated within the common
memory-allowing parallel class $\mathcal S_Q$, while
$\mathcal S_{C,1\to2}$ and $\mathcal S_{Q,1\to2}$ are treated within the
common memory-allowing sequential class $\mathcal S_{Q,1\to2}$.  We do not
attempt to distinguish classical from quantum memory.

The experimental MP settings use a reduced-randomization scheme: in settings 2–7, only Alice's and Bob's measurement bases and Alice's second preparation state are randomly selected, while the remaining preparation states are fixed. This scheme has the same pattern of fixed and randomized components as the partial-randomization setting of Theorem~\ref{thm:partial-randomization}.
After each setting is selected, all its
measurement and preparation operations are kept fixed throughout the repeated
trials.  The finite counts are then analyzed by the prescribed $\chi^2$
tests for the Markov conditions in Eq.~\eqref{HI4}, at significance level
$\alpha=0.05$.

As depicted in Fig.~\ref{fig:setup}, our experimental demonstration utilizes heralded single photons generated via spontaneous parametric down-conversion (SPDC) in a type-II PPKTP crystal, with quantum states encoded in hybrid polarization and path degrees of freedom (DOFs). A specific ``DOF switch'' allows Charlie to select between parallel and sequential strategies. This configuration permits the second player to operate effectively across different DOFs without altering the optical hardware. Furthermore, the operational order of the players is exchanged by switching the propagation direction (clockwise vs. counter-clockwise) via mirrors M1--M3. Crucially, for any strategy, Alice and Bob implement the measurement-and-preparation process using identical settings and detection modules; thus, they are unable to infer the underlying order from their local configurations. The detailed information about 
the experimental setup and the MP channels is in the Method section. 

\begin{figure}[h]
    \centering
    \includegraphics[width=\columnwidth]{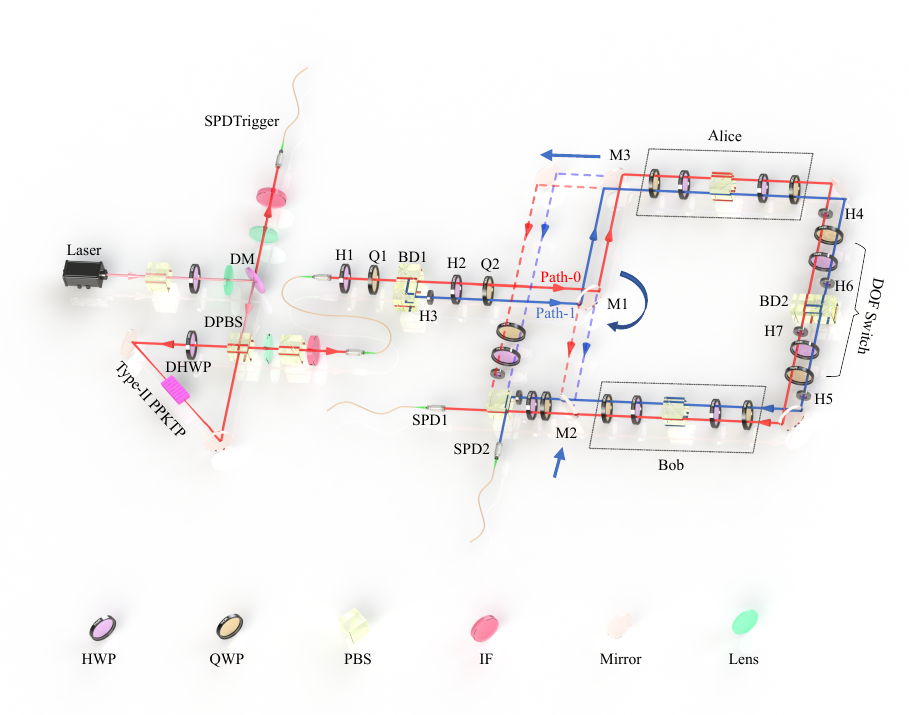}
    \caption{Experimental setup. Solid and dashed lines represent photon propagation directions of Alice $\rightarrow$ Bob and Bob $\rightarrow$ Alice, respectively. Detailed configurations for different strategies are provided in the Supplemental Materials.}
    \label{fig:setup}
\end{figure}
For every implemented process realization and each of the seven MP
settings, the Markov-condition-based procedure returned the expected
class in the hierarchy~\eqref{HI1}, as summarized in
Table~\ref{tab:classification-results}.  This experimental result is consistent with our analysis: even with fewer randomized components than condition~(S2), the reduced-randomness setting retained sufficient identifying information to classify all tested processes correctly, at the tested significance level. The
detailed $\chi^2$ values are reported in the Supplemental Materials.

These finite-sample successes do not establish testing power or a guarantee for arbitrary process classes. Theorem~\ref{thm:partial-randomization} supplies the rigorous guarantee separately from the data---an exact-distribution probability-one result for the stated classes under its stated sampling assumptions.

\begin{table*}[htbp]
\centering
\caption{Experimental results of the Markov-condition-based
classification for the seven MP settings.  For every implemented process
realization and every setting, the finite-sample $\chi^2$ tests yield the
expected class in the hierarchy~\eqref{HI1}.  The expected classes for the
$\mathcal S_C$ and $\mathcal S_Q$ realizations are both $\mathcal S_Q$,
whereas those for the $\mathcal S_{C,1\to2}$ and
$\mathcal S_{Q,1\to2}$ realizations are both
$\mathcal S_{Q,1\to2}$.}
\label{tab:classification-results}
\resizebox{\textwidth}{!}{%
\begin{tabular}{lcccccccc}
\hline
\textbf{Implemented realization}
& \textbf{Expected class}
& \textbf{Set. 1}
& \textbf{Set. 2}
& \textbf{Set. 3}
& \textbf{Set. 4}
& \textbf{Set. 5}
& \textbf{Set. 6}
& \textbf{Set. 7} \\
\hline
$\mathcal S_I$
& $\mathcal S_I$
& Success & Success & Success & Success & Success & Success & Success \\
$\mathcal S_C$
& $\mathcal S_Q$
& Success & Success & Success & Success & Success & Success & Success \\
$\mathcal S_Q$
& $\mathcal S_Q$
& Success & Success & Success & Success & Success & Success & Success \\
$\mathcal S_{N,1\to2}$
& $\mathcal S_{N,1\to2}$
& Success & Success & Success & Success & Success & Success & Success \\
$\mathcal S_{C,1\to2}$
& $\mathcal S_{Q,1\to2}$
& Success & Success & Success & Success & Success & Success & Success \\
$\mathcal S_{Q,1\to2}$
& $\mathcal S_{Q,1\to2}$
& Success & Success & Success & Success & Success & Success & Success \\
\hline
\end{tabular}%
}
\end{table*}
\section*{DISCUSSION}
In this work, we have proposed a characterization-free protocol for the identification of several definite-order strategies mediating two quantum channels, which does not require explicit calibrated operator-level descriptions of the players' devices under the stated assumptions.
Under tomographically complete settings, we established a one-to-one
correspondence between the relevant strategy classes and the Markov
conditions of the observed statistics.  We then proved that, under the full
componentwise randomization condition~(S2), a single binary
measure-and-prepare setting retains exact-distribution identifiability with
probability one.

The optical experiment used a reduced-randomness construction: Alice's
and Bob's measurement bases and one of Alice's preparation states were varied,
while the remaining preparation states were kept fixed.  Nevertheless, the
Markov-condition-based procedure returned the expected causal-order and
memory-presence class for every tested process-setting pair.

Theorem~\ref{thm:partial-randomization} establishes generic exact-distribution identification for the specified parallel and $1\to2$ classes under reduced sampling. The analogous result for arbitrary classes, and minimal randomization, remain open.  No
distinction between classical and quantum memory is attempted in the present
experimental study.
Our results offer an experimentally feasible framework for characterizing unknown causal structures.
It is scalable in principle, thus could be useful for studying the causal structures in quantum networks or large-scale quantum processors when it is extended into multi-player settings. 
Also, the characterization-free property of our framework
makes it possible to study the causal structure 
of more general frameworks such as the general probabilistic theory \cite{plavala2023general,janotta2014generalized,yu2024measurement}. 

Another direction is to characterize minimal partial-randomization assumptions beyond the classes covered by Theorem~\ref{thm:partial-randomization}.

\section*{METHODS}
\subsection*{Testing the Markov chain condition}
In this part we explain how to do the hypothesis testing on the Markov chain condition using $\chi^{2}$ test. 
Note that all Markov chain condition we need to test in our case can be reduced to a three variable case,
e.g., testing $J_1-K_1-J_2-K_2$ is equivalent to 
testing $J_1-K_1-J_2$ and $(J_1,K_1)-J_2-K_2$.
Therefore, we just consider the three variables case here. 
We let $X$, $Y$, $Z$ be three random variables with $d_X$, $d_Y$, $d_Z$ outcomes respectively. We consider testing the Markov chain condition $X-Y-Z$ w.l.o.g.
Then we can formulate the hypothesis as follows:
\begin{itemize}
	\item Null Hypothesis $H_{0}$: $X$, $Y$, $Z$ form a Markov chain $X-Y-Z$ ($X$ and $Z$ are conditionally independent given $Y$).
	\item Alternative Hypothesis $H_{1}$:  $X$, $Y$, $Z$ does not form a Markov chain ($X$ and $Z$ are not conditionally independent given $Y$).
\end{itemize}

Suppose we make $N$ repetitive experiments, obtaining $N$ value triplets $(X,Y,Z)$. Let $N_{ijk}$ be the number of counts that condition $(X=i,Y=j,Z=k)$ is satisfied. 
Using the fact that $X,Y,Z$ form a Markov chain $X-Y-Z$ if and only if $X$ and $Z$ are conditionally independent given Y, we construct the expected counts $E_{ijk}$ when $X,Y,Z$ form a Markov chain as 
\begin{equation}
	E_{ijk}=N_{j}\frac{N_{ij}}{N_{j}}\frac{N_{jk}}{N_{j}}=\frac{N_{ij}N_{jk}}{N_{j}},
\end{equation}
where $N_{ij}=\sum_{k}N_{ijk}$, $N_{jk}=\sum_{i}N_{ijk}$, and
$N_j=\sum_{ik}N_{ijk}$.  A stratum with $N_j=0$ contains no
observations and is omitted from the statistic.  Within a retained stratum,
a cell with $E_{ijk}=0$ necessarily has $N_{ijk}=0$ under the fitted marginal
model and is assigned contribution zero.  The Pearson statistic is therefore
\begin{equation}
\chi^{2}=\sum_{j:N_j>0}\ \sum_{(i,k):E_{ijk}>0}
\frac{(N_{ijk}-E_{ijk})^2}{E_{ijk}}.
\end{equation}
And the degree of freedom of the test is 
\begin{equation}
	df=(d_{X}-1)d_{Y}(d_{Z}-1).
\end{equation}

Then we compare the obtained $\chi^{2}$ value with the critical value
$\chi^{2}_{crit}$ for the chosen significance level.  We do not reject
$H_0$ when $\chi^{2}\leq\chi^{2}_{crit}$ and reject it otherwise.  This
Pearson calibration is an asymptotic approximation and is reliable only when
the retained expected counts are sufficiently large.  If sparse expected
counts occur, an exact conditional procedure or a Monte Carlo calibration of
the null distribution should be used instead.

\subsection*{Experimental setup}

In this part we introduce more details of our experimental setup,
the parameters to realize Charlie's strategies, and the settings of MP channels.

As depicted in Fig.~2, our experimental demonstration is implemented on a photonic platform where the quantum state is encoded onto single photons generated via the spontaneous parametric down-conversion (SPDC) process. A horizontally polarized 775 nm continuous-wave (CW) laser pumps a 20-mm-long type-II PPKTP crystal, generating degenerate photon pairs at 1550 nm ($\ket{H}_p \to \ket{H}_s\ket{V}_i$). Following the crystal, the photon pairs undergo a polarization rotation to the state $\ket{V}_s\ket{H}_i$ via a dual-wavelength half-wave plate (dHWP) set at $45^\circ$. They are subsequently separated by a dual-wavelength polarizing beam splitter (dPBS). The vertically polarized photon is detected by a single-photon detector (SPD) to serve as the heralding trigger. The conjugate horizontally polarized photon is purified by a PBS, coupled into a standard single-mode fiber (SMF), and serves as the heralded single-photon source for the subsequent task. The 40 nm bandwidth filters are inserted before the fiber couplers to remove residual pump light. 

We encode the qubit on the hybrid polarization and path degrees of freedom (DOFs). Efficient coupling between these DOFs is achieved using a Beam Displacer (BD), which acts as a polarization-dependent spatial displacer. Specifically, the BD transmits $\ket{H}$ photons directly into spatial mode $\ket{0}_p$, while inducing a 4-mm lateral displacement for $\ket{V}$ photons into spatial mode $\ket{1}_p$. This hybrid structure allows Charlie to flexibly switch between strategies. For the sequential strategy, we utilize only the polarization DOF within a single spatial mode ($\ket{1}_p$). Conversely, for the parallel strategy, hyper encoded input states are prepared with both DOFs. To achieve this, we employ a versatile state preparation stage. First, wave plates H1 and Q1 prepare an initial polarization superposition $\ket{\Phi_{\text{in}}} = (\cos\theta_1\ket{H} + e^{i\phi_1}\sin\theta_1\ket{V})\ket{0}_p$. A BD then maps this superposition to the path DOF. Crucially, a half-wave plate (H3) is inserted specifically into path-1 to implement a controlled-unitary operation $U(\theta_2)$. This transforms the state into $\ket{\Phi'_{\text{in}}} = \cos\theta_1\ket{H}\ket{0}_p + e^{i\phi_1}U(\theta_2)\ket{V}\ket{1}_p$. By adjusting H3, either separable or entangled states can be prepared. Finally, additional wave plates (H2 and Q2) set the global polarization. To realize the probability preparation $\rho_{1,\lambda}$ required for strategy $\mathcal{S}_{C,1\rightarrow 2}$, H2 and Q2 are mounted on Electric Rotation Mounts (ERMs), allowing for precise, automated control of the classical mixing parameter $\lambda$.  

In the classification process, the two players, Alice and Bob, operate under strict isolation from the environment. In the sequential strategy, operations are straightforwardly performed on the polarization DOF using a standard QWP-HWP-PBS measurement structure. Following measurement, a subsequent set of HWP and QWP mounted on ERMs is used to prepare the measurement-outcome-dependent state $\rho_{k_1}$. In the parallel strategy, the implementation requires independent handling of the two DOFs. Alice performs her operations on the polarization DOF using optical components with large apertures that span all spatial paths. This design ensures her operations are path-independent. Following Alice's stage, a bi-directional "DOF switch" is introduced. In its active configuration (used for parallel strategy), this structure maps the path information onto the polarization DOF, transforming the path-encoded state into a polarization-encoded state. In its inactive configuration, it acts as an identity operation. This switch enables Bob to analyze the original path state using the same polarization-based settings as Alice, thereby effectively realizing a non-interfering parallel measurement scheme. 

A specific requirement arises in the sequential strategy $\mathcal{S}_{Q,1\rightarrow 2}$, where the path DOF serves as an auxiliary qubit to implement a controlled-NOT (CNOT) operation on Alice's output state. This is realized by inserting a $45^\circ$-rotated HWP (labeled H4 or H5) into path-1, which flips the polarization conditional on the photon's presence in that path.To recover the target statistics and ensure that Alice and Bob cannot infer the underlying order strategy, we employ an identical path-independent detection scheme across all strategies. Experimentally, we map the two spatial paths to orthogonal polarization modes using a bit-flip operation (a $45^\circ$ HWP in path-1) followed by a final PBS. The photons are then directed to two separate detectors (SPD1 and SPD2). The final input-output statistics are obtained by summing the coincidence counts of SPD1 $\&$ Trigger with SPD2 $\&$ Trigger. This summation effectively traces out the path degree of freedom, ensuring a consistent and strategy-agnostic measurement interface.

\section*{DATA AVAILABILITY}
The experimental data produced in this study are available in the repository \href{https://github.com/BC-YU/Char-Free-Identification}{https://github.com/BC-YU/Char-Free-Identification}.

\section*{CODE AVAILABILITY}
The code used to produce the results in this study is available in the repository \href{https://github.com/BC-YU/Char-Free-Identification}{https://github.com/BC-YU/Char-Free-Identification}.

\section*{ACKNOWLEDGEMENT}
M.H. and B.Y. were supported in part by
the Guangdong Provincial Quantum Science Strategic Initiative (Grant No. GDZX2505003).
M.H. was supported in part by
the General R\&D Projects of 1+1+1 CUHK-CUHK(SZ)-GDST Joint Collaboration Fund (Grant No. GRDP2025-022), 
the Shenzhen International Quantum Academy (Grant No. SIQA2025KFK\linebreak 
T07), and the National Natural Science Foundation of China under Grant 62171212.
Y.Z. is supported by the National Natural Science Foundation of China (Grant No. 12574401), the Guangdong Project (Grant No. 2024TQ08A680), the Guangdong Provincial Quantum Science Strategic Initiative (GDZX2403005, GDZX2403001, GDZX2506001, GDZX2503002 and GDZX2403002).
B.Y. was supported by the GuangDong Basic and Applied Basic Research Foundation (Grant No. 2024A1515012405).

\section*{AUTHOR CONTRIBUTIONS}
M.H. initiated the study. M.H. and B.Y. developed the theoretical framework. L.C., B.Y. and Y.Z. performed the experiment and processed the experimental data. M.H., B.Y. and Y.Z. contributed to the writing of the manuscript. 

\section*{COMPETING INTERESTS}
The authors declare no competing interests.

\bibliography{ref}
\bibliographystyle{unsrt}

\section*{Supplemental Material}

\renewcommand{\thesection}{S\arabic{section}}
\renewcommand{\theequation}{S\arabic{equation}}
\renewcommand{\thefigure}{S\arabic{figure}}
\renewcommand{\thetable}{S\arabic{table}}

\setcounter{section}{0}
\setcounter{equation}{0}
\setcounter{figure}{0}
\setcounter{table}{0}

\section{Operational assumptions and the meaning of \\
``characterization-free''}

Throughout this work, \emph{characterization-free} means that the inference does not require an explicit calibrated operator-level description of the parties' devices (e.g., matrix representations, eigenbases, or an implementation model of the POVM and preparation operators). 
Instead, our guarantees rely on explicit \emph{operational non-degeneracy} assumptions on the accessible settings, stated as (S1) and (S2) in the main text.
We emphasize that ``characterization-free'' is not the same as device-independence in the usual sense, but only refers to not assuming calibrated operator descriptions.

More specifically, condition (S1) assumes that the chosen measurement and preparation operators span the relevant local Hermitian-operator spaces. 
Operationally, this presumes enough calibration/coordination to speak about spanning and linear independence of the implemented settings, while not requiring their explicit calibrated forms.
Under (S1), Theorem~1 yields an if-and-only-if identification criterion, but requires 
$d_{(I,1)}^2 d_{(I,2)}^2 d_{(O,1)}^2 d_{(O,2)}^2$ settings to generate tomographically complete correlations.

On the other hand, condition (S2) concerns independently sampled binary measurements and preparation states.  Each sampling law is required only to be absolutely continuous with respect to a unitarily invariant reference measure that gives zero probability to scalar operators.  The Supplemental proof derives the spanning and nonconcentration properties needed for Theorem~2 from these assumptions.  Haar-random rank-one projectors and Haar-random pure states are included.
\section{Proof of Theorem 1}\label{apdxA}
For the relation between constraints of process matrices
and the Markovian conditions for the distribution $P_{J_1,J_2,K_1,K_2}$,
we have the following equivalence conditions.

\begin{proofof}{Theorem 1}
	
	\noindent\textbf{Step 1: One-way quantum class.}\quad
	We show that
		condition~(9) is equivalent to
		$(J_1,K_1)-J_2-K_2$.  The reverse direction is obtained by exchanging
		labels $1$ and $2$.
		
		For brevity, put $a=(j_1,k_1)$, $b=j_2$, $c=k_2$, and
		$r_2(c|b):=P_{K_2|J_2}(c|b)$.  Define
		\begin{equation}
			Q(a,b,c):=\Tr\!\left[W_S\left(C[\Gamma_{1,a}]\otimes
			M_{2,b}^{T}\otimes\rho_c\right)\right].
			\label{eq:step1-Q}
		\end{equation}
		The generalized Born rule gives the joint-probability identity
		\begin{equation}
			p(a,b,c)=r_2(c|b)Q(a,b,c).
			\label{eq:step1-local-factorization}
		\end{equation}
		Moreover, summing over $a$ and using the trace-preserving property of
		Alice's total instrument gives
		\begin{equation}
			p(b,c)=r_2(c|b)p(b).
			\label{eq:step1-local-marginal}
		\end{equation}
		
		If condition~(9) holds, then $Q(a,b,c)$ is independent of $c$ and equals
		\begin{equation}
			Q(a,b):=\frac{1}{d_{(O,2)}}
			\Tr\!\left[(\Tr_{(O,2)}W_S)
			(C[\Gamma_{1,a}]\otimes M_{2,b}^{T})\right].
			\label{eq:step1-Q-independent}
		\end{equation}
		Equations~\eqref{eq:step1-local-factorization} and
		\eqref{eq:step1-Q-independent} yield
		$p(a,b,c)=r_2(c|b)Q(a,b)$.  Summing over $c$ gives
		$Q(a,b)=p(a,b)$, and hence
		\begin{equation}
			p(a,b,c)=p(a,b)r_2(c|b).
			\label{eq:step1-Markov-factorization}
		\end{equation}
		Together with Eq.~\eqref{eq:step1-local-marginal}, this is the
		joint-distribution factorization for $(J_1,K_1)-J_2-K_2$, including the
		case $p(b)=0$.
		
		Conversely, assume $(J_1,K_1)-J_2-K_2$.  Since $K_2$ is generated from
		$J_2$ by the fixed local kernel $r_2$, the Markov factorization is exactly
		Eq.~\eqref{eq:step1-Markov-factorization}.  Comparing it with
		Eq.~\eqref{eq:step1-local-factorization} and cancelling only the strictly
		positive factor $r_2(c|b)$ gives
		\begin{equation}
			Q(a,b,c)=p(a,b)
			\quad\text{for every }a,b,c.
			\label{eq:step1-Q-Markov}
		\end{equation}
		Thus the pairing in Eq.~\eqref{eq:step1-Q} is independent of $c$.
		To obtain an operator identity, set
		$D:=W_S-d_{(O,2)}^{-1}(\Tr_{(O,2)}W_S)\otimes I_{(O,2)}$.
		For fixed $a,b$, let
		$F_{a,b}(Y):=\Tr[W_S(C[\Gamma_{1,a}]\otimes M_{2,b}^{T}\otimes Y)]$.
		Equation~\eqref{eq:step1-Q-Markov} gives
		$F_{a,b}(\rho_c)=p(a,b)$ for every $c$.  Since every $\rho_c$
		has trace one and the $\rho_c$ span the Hermitian operators on $O_2$,
		linearity implies $F_{a,b}(Y)=p(a,b)\Tr Y$ for every Hermitian $Y$.
		In particular, $F_{a,b}(I_{(O,2)}/d_{(O,2)})=p(a,b)$.
		The latter pairing is
		$d_{(O,2)}^{-1}\Tr[(\Tr_{(O,2)}W_S)
		(C[\Gamma_{1,a}]\otimes M_{2,b}^{T})]$.
		Thus $\Tr[D(C[\Gamma_{1,a}]\otimes M_{2,b}^{T}\otimes Y)]=0$
		for all $a,b$ and all Hermitian $Y$.
		Under (S1) and the strictly positive preparation kernels,
		$\{C[\Gamma_{1,a}]\}_a$ spans the Hermitian operators on $I_1O_1$;
		$\{M_{2,b}^{T}\}_b$ spans those on $I_2$.  Their products with
		Hermitian $Y$ therefore span the full Hermitian-operator space.
		Hilbert--Schmidt nondegeneracy gives $D=0$, that is,
		\begin{equation}
			W_S=\frac{1}{d_{(O,2)}}(\Tr_{(O,2)}W_S)\otimes I_{(O,2)},
		\end{equation}
		which is condition~(9).  No conditional probability on a possibly
		zero-probability conditioning event is used.
	
	\noindent\textbf{Step 2: A factorization criterion for one-way memoryless class.}\quad
	Define the normalized factors
		\begin{align}
			\rho_1&:=\frac{1}{d_{(O,1)}d_{(O,2)}}
			\Tr_{(O,1)(I,2)(O,2)}W_S,
			\label{def-rho1-factorization-new}\\
			C&:=\frac{1}{d_{(O,2)}}\Tr_{(I,1)(O,2)}W_S.
			\label{def-C-factorization-new}
		\end{align}
		We show that the following identity implies condition~(7):
		\begin{align}
			&\Tr\!\left[W_S
			(M_{1,j_1}^{T}\otimes\rho_{k_1}\otimes
			M_{2,j_2}^{T}\otimes\rho_{k_2})\right]\notag\\
			&\quad=\Tr(\rho_1M_{1,j_1}^{T})
			\Tr\!\left[C(\rho_{k_1}\otimes M_{2,j_2}^{T})\right].
			\label{eq:normalized-factorization-identity}
		\end{align}
		This identity remains meaningful when either factor on the right-hand side
		vanishes.  It is equivalently the equality of pairings
		\begin{align}
			&\Tr\!\left[W_S
			(M_{1,j_1}^{T}\otimes\rho_{k_1}\otimes
			M_{2,j_2}^{T}\otimes\rho_{k_2})\right]\notag\\
			&\quad=\Tr\!\left[(\rho_1\otimes C\otimes I_{(O,2)})
			(M_{1,j_1}^{T}\otimes\rho_{k_1}\otimes
			M_{2,j_2}^{T}\otimes\rho_{k_2})\right].
			\label{eq:factorization-product-pairing}
		\end{align}
		By (S1), the measurement effects and preparation states span their
		respective Hermitian-operator spaces.  Transposition is an invertible real
		linear map on Hermitian operators, so the transposed measurement effects
		also span their spaces.  The product operators in
		Eq.~\eqref{eq:factorization-product-pairing} therefore span the full
		Hermitian-operator space.  Hilbert--Schmidt nondegeneracy gives
		\begin{equation}
			W_S=\rho_1\otimes C\otimes I_{(O,2)}.
			\label{eq:factorization-product-result}
		\end{equation}
		The process normalization $\Tr W_S=d_{(O,1)}d_{(O,2)}$ gives
		$\Tr\rho_1=1$.  Positivity of $\rho_1$ and $C$ follows from positivity of
		partial traces.  Substituting Eq.~\eqref{eq:factorization-product-result}
		into the one-way process normalization gives
		$\Tr_{(I,2)}C=I_{(O,1)}$.  Thus $C=C[\Lambda_C]$ for a CPTP map, and
		condition~(7) holds.
	
	\noindent\textbf{Step 3: One-way memoryless class.}\quad
	We show that condition~(7) is equivalent to
		$J_1-K_1-J_2-K_2$.  The reverse direction follows by exchanging labels
		$1$ and $2$.
		
		If condition~(7) holds, the generalized Born rule directly factorizes as
		\begin{align}
			&p(j_1,k_1,j_2,k_2)\notag\\
			&\quad=r_1(k_1|j_1)r_2(k_2|j_2)
			\Tr(\rho_1M_{1,j_1}^{T})
			\Tr\!\left[C[\Lambda_C](\rho_{k_1}\otimes M_{2,j_2}^{T})\right],
			\label{eq:step3-forward-factorization}
		\end{align}
		which is the joint-distribution factorization for
		$J_1-K_1-J_2-K_2$.
		
		Conversely, assume $J_1-K_1-J_2-K_2$.  Write
		$r_i(k_i|j_i):=P_{K_i|J_i}(k_i|j_i)$ and
		\begin{equation}
			q(j_1,k_1,j_2,k_2):=
			\Tr\!\left[W_S(M_{1,j_1}^{T}\otimes\rho_{k_1}\otimes
			M_{2,j_2}^{T}\otimes\rho_{k_2})\right].
			\label{eq:step3-q}
		\end{equation}
		The generalized Born rule gives
		$p=r_1r_2q$.  On the other hand, the Markov-chain factorization gives
		\begin{equation}
			p(j_1,k_1,j_2,k_2)=P_{J_1}(j_1)r_1(k_1|j_1)
			P_{J_2|K_1}(j_2|k_1)r_2(k_2|j_2).
			\label{eq:step3-Markov-factorization}
		\end{equation}
		Because $r_1$ and $r_2$ are strictly positive, comparison yields
		\begin{equation}
			q(j_1,k_1,j_2,k_2)=P_{J_1}(j_1)P_{J_2|K_1}(j_2|k_1).
			\label{eq:step3-q-factorization}
		\end{equation}
		Strict positivity of $r_1$ implies $P_{K_1}(k_1)>0$ for every $k_1$, so
		$P_{J_2|K_1}$ in Eq.~\eqref{eq:step3-Markov-factorization} is well defined.
		The assumed chain also implies $(J_1,K_1)-J_2-K_2$.  By Step~1,
		$W_S$ satisfies the one-way condition~(9); its one-way process
		normalization may therefore be used in the following marginal calculation.
		Summing
		Eq.~\eqref{eq:step3-q-factorization} over $j_1$ gives
		\begin{equation}
			\sum_{j_1'}q(j_1',k_1,j_2,k_2)=P_{J_2|K_1}(j_2|k_1),
		\end{equation}
		and hence the identity
		\begin{equation}
			q(j_1,k_1,j_2,k_2)=P_{J_1}(j_1)
			\sum_{j_1'}q(j_1',k_1,j_2,k_2).
			\label{eq:step3-product-identity}
		\end{equation}
		By summing the generalized Born rule over all outcomes except $J_1$,
		and using the one-way process normalization just established, we obtain
		\begin{equation}
			P_{J_1}(j_1)=\Tr(\rho_1M_{1,j_1}^{T}),
			\label{eq:PJ1-rho1}
		\end{equation}
		where $\rho_1$ is the normalized partial trace in
		Eq.~\eqref{def-rho1-factorization-new}.  The one-way condition makes
		$q$ independent of $k_2$.  Using
		Eq.~\eqref{def-C-factorization-new} and POVM completeness,
		\begin{align}
			\sum_{j_1'}q(j_1',k_1,j_2,k_2)
			&=\frac{1}{d_{(O,2)}}
			\Tr\!\left[(\Tr_{(I,1)(O,2)}W_S)
			(\rho_{k_1}\otimes M_{2,j_2}^{T})\right]\notag\\
			&=\Tr\!\left[C(\rho_{k_1}\otimes M_{2,j_2}^{T})\right].
			\label{eq:step3-sum-C}
		\end{align}
		Substituting Eqs.~\eqref{eq:PJ1-rho1} and
		\eqref{eq:step3-sum-C} into
		Eq.~\eqref{eq:step3-product-identity} gives exactly
		Eq.~\eqref{eq:normalized-factorization-identity}.  Step~2 now implies
		condition~(7).
		Thus condition~(7) is equivalent to $J_1-K_1-J_2-K_2$.
	
	\noindent\textbf{Step 4: Process-class intersection identities.}\quad
	We first establish the two intersection identities used below.
		Imposing the two one-way quantum conditions $(\mathcal{S}_{Q,1\to 2}, \mathcal{S}_{Q,2\to 1})$ gives
		\begin{align}
			W_{\mathcal S}
			&=\frac{1}{d_{(O,2)}}(\Tr_{(O,2)}W_{\mathcal S})\otimes I_{(O,2)},\notag\\
			W_{\mathcal S}
			&=\frac{1}{d_{(O,1)}}(\Tr_{(O,1)}W_{\mathcal S})\otimes I_{(O,1)}.
		\end{align}
		Applying the two commuting replacement maps successively yields
		\begin{equation}
			W_{\mathcal S}=\frac{1}{d_{(O,1)}d_{(O,2)}}
			(\Tr_{(O,1)(O,2)}W_{\mathcal S})\otimes I_{(O,1)}\otimes I_{(O,2)}.
		\end{equation}
		Hence, with
		$\rho_{12}:=(d_{(O,1)}d_{(O,2)})^{-1}
		\Tr_{(O,1)(O,2)}W_{\mathcal S}$, the process has exactly the form in
		condition~(6).  The converse follows immediately by taking the indicated
		partial traces.  Therefore,
		\begin{equation}
			\mathcal S_Q=\mathcal S_{Q,1\to2}\cap\mathcal S_{Q,2\to1}.
			\label{eq:SQ-intersection-TH1}
		\end{equation}
		
		Likewise, imposing the two opposite memoryless sequential factorizations gives
		\begin{align}
			W_{\mathcal S}&=\rho_1\otimes C_{O_1I_2}\otimes I_{(O,2)},\notag\\
			W_{\mathcal S}&=\rho_2\otimes C_{O_2I_1}\otimes I_{(O,1)}.
		\end{align}
		Taking $d_{(O,2)}^{-1}\Tr_{(O,2)}$ of the two factorizations gives
		\begin{equation}
			\rho_1\otimes C_{O_1I_2}
			=\frac{1}{d_{(O,2)}}(\Tr_{(O,2)}C_{O_2I_1})
			\otimes I_{(O,1)}\otimes\rho_2.
			\label{eq:two-memoryless-reductions}
		\end{equation}
		Tracing Eq.~\eqref{eq:two-memoryless-reductions} over $O_1I_2$ and using
		$\Tr C_{O_1I_2}=d_{(O,1)}$ yields
		$d_{(O,2)}^{-1}\Tr_{(O,2)}C_{O_2I_1}=\rho_1$.
		Substitution back into Eq.~\eqref{eq:two-memoryless-reductions}, followed
		by tracing over $I_1$, gives
		$C_{O_1I_2}=I_{(O,1)}\otimes\rho_2$.  Therefore,
		\begin{equation}
			W_{\mathcal S}=\rho_1\otimes\rho_2\otimes I_{(O,1)}\otimes I_{(O,2)},
		\end{equation}
		which is condition~(4).  The converse inclusion is immediate, and hence
		\begin{equation}
			\mathcal S_I=\mathcal S_{N,1\to2}\cap\mathcal S_{N,2\to1}.
			\label{eq:SI-intersection-TH1}
	\end{equation}
	
	\noindent\textbf{Step 5: Characterization of the quantum parallel class.}\quad
	We show that condition~(6) is equivalent to the Markov chain
		$K_1-J_1-J_2-K_2$.
		
		Suppose first that condition~(6) holds, so that
		$W_{\mathcal S}=\rho_{12}\otimes I_{(O,1)}\otimes I_{(O,2)}$.
		The generalized Born rule gives
		\begin{align}
			&P_{J_1,K_1,J_2,K_2}(j_1,k_1,j_2,k_2)\notag\\
			&\quad=P_{K_1|J_1}(k_1|j_1)P_{K_2|J_2}(k_2|j_2)
			\Tr\!\left[\rho_{12}(M_{1,j_1}^{T}\otimes M_{2,j_2}^{T})\right].
			\label{eq:SQ-direct-factorization}
		\end{align}
		Defining
		$P_{J_1,J_2}(j_1,j_2):=
		\Tr[\rho_{12}(M_{1,j_1}^{T}\otimes M_{2,j_2}^{T})]$, this becomes
		\begin{equation}
			P_{J_1,K_1,J_2,K_2}
			=P_{J_1,J_2}P_{K_1|J_1}P_{K_2|J_2},
		\end{equation}
		which is the Markov chain $K_1-J_1-J_2-K_2$.
		
		Conversely, assume that $K_1-J_1-J_2-K_2$ holds.  By the factorization
		defining this four-variable Markov chain, it implies both
		$(J_1,K_1)-J_2-K_2$ and $K_1-J_1-(J_2,K_2)$.  Step~1 and its counterpart
		obtained by exchanging labels $1$ and $2$ therefore give
		$W_{\mathcal S}\in\mathcal S_{Q,1\to2}$ and
		$W_{\mathcal S}\in\mathcal S_{Q,2\to1}$, respectively.  Equation
		~\eqref{eq:SQ-intersection-TH1} then gives
		$W_{\mathcal S}\in\mathcal S_Q$, which is condition~(6).
	
	\noindent\textbf{Step 6: Characterization of the individual class.}\quad
	We show that condition~(4) is equivalent to
		$(J_1,K_1)\perp(J_2,K_2)$.
		
		Suppose first that condition~(4) holds, so that
		$W_{\mathcal S}=\rho_1\otimes\rho_2\otimes I_{(O,1)}\otimes I_{(O,2)}$.
		The generalized Born rule gives
		\begin{align}
			&P_{J_1,K_1,J_2,K_2}(j_1,k_1,j_2,k_2)\notag\\
			&\quad=P_{K_1|J_1}(k_1|j_1)\Tr(\rho_1M_{1,j_1}^{T})
			P_{K_2|J_2}(k_2|j_2)\Tr(\rho_2M_{2,j_2}^{T})\notag\\
			&\quad=P_{J_1,K_1}(j_1,k_1)P_{J_2,K_2}(j_2,k_2).
			\label{eq:SI-direct-factorization}
		\end{align}
		Thus, $(J_1,K_1)\perp(J_2,K_2)$.
		
		Conversely, assume that $(J_1,K_1)\perp(J_2,K_2)$.  The corresponding
		factorization implies both directional Markov chains
		$J_1-K_1-J_2-K_2$ and $K_1-J_1-K_2-J_2$.  Step~3 and its counterpart
		obtained by exchanging labels $1$ and $2$ therefore give
		$W_{\mathcal S}\in\mathcal S_{N,1\to2}$ and
		$W_{\mathcal S}\in\mathcal S_{N,2\to1}$, respectively.  Equation
		~\eqref{eq:SI-intersection-TH1} then gives
		$W_{\mathcal S}\in\mathcal S_I$, which is condition~(4).
\end{proofof}

Combining the above methods, we are able to identify all strategy classes in (14).

\section{Proof of Theorem 2}\label{apdxC}
\subsection{Statement of Theorem 2}

We first state explicitly the sampling condition used in this proof.

\begin{description}
	\item[(S2)]
	The six operators $M_1,M_2,\rho_{1,1},\rho_{1,2},\rho_{2,1},\rho_{2,2}$
	are sampled independently. We collect them into the random tuple
	\begin{equation}
		\boldsymbol{\Xi}:=\left(M_1,M_2,\rho_{1,1},\rho_{1,2},\rho_{2,1},\rho_{2,2}\right).
		\label{eq:random-operator-tuple}
	\end{equation}
	For each component, its sampling law is absolutely continuous with respect
	to a unitarily invariant probability measure on the corresponding set of
	positive operators.  This reference measure assigns probability zero to
	operators proportional to the identity.  For a measurement component the
	operator set is $\{M:0\leq M\leq I\}$, whereas for a preparation component
	it is the set of density operators on the corresponding Hilbert space.
	Alice and Bob use the binary POVMs $\{M_i,I_{I,i}-M_i\}$ and the two
	preparation states $\rho_{i,1},\rho_{i,2}$. Hence $J_i,K_i\in\{1,2\}$.
	The conditional preparation probabilities satisfy
	$P_{K_i|J_i}(k_i|j_i)>0$ for all $j_i,k_i$ and $i=1,2$.
\end{description}

For convenience, we now restate Theorem~2.
\begin{theorem}[Theorem 2]
	Assume that Charlie's strategy is represented by a fixed physical process
	matrix $W_{\mathcal S}$ that is independent of Alice's and Bob's random
	choices, and that Alice and Bob choose their local operations according to
	condition~(S2). In particular, the variables in the induced distribution
	are binary: $J_i,K_i\in\{1,2\}$ for $i=1,2$, where $J_i$ labels the
	outcome of the POVM $\{M_i,I_{I,i}-M_i\}$ and $K_i$ labels the state
	chosen from $\{\rho_{i,1},\rho_{i,2}\}$.
	
	Let $\mathcal S_X$ be one of $\mathcal S_I$, $\mathcal S_Q$,
	$\mathcal S_{N,1\to2}$, $\mathcal S_{N,2\to1}$,
	$\mathcal S_{Q,1\to2}$, and $\mathcal S_{Q,2\to1}$. If
	$W_{\mathcal S}\notin\mathcal S_X$, then the induced distribution of
	$(J_1,K_1,J_2,K_2)$ does not satisfy $Mar(\mathcal S_X)$ with probability
	one.
\end{theorem}
The probability in this theorem is taken over the sampling of the MP
setting. Condition~(S2) is a sufficient sampling condition for this
probability-one conclusion and is not claimed to be necessary. The theorem
concerns the exact induced distribution and does not provide a finite-sample
power or sample-complexity bound.
\subsection{Proofs of Theorem~2 for $\mathcal S_{Q,1\to2}$, $\mathcal S_{Q,2\to1}$, 
	$\mathcal S_{N,1\to2}$
	and
	$\mathcal S_{N,2\to1}$}
We use the real vector space $\mathsf{Herm}(\mathcal H)$ of Hermitian
operators on $\mathcal H$.  We first derive the consequences of condition~(S2)
that are used below.

\begin{lemma}[Unitary-orbit spanning]
	\label{lem:unitary-orbit-spanning}
	Let $X_0$ be a nonzero positive operator that is not proportional to the
	identity.  Then the real linear span of
	\[
	\{UX_0U^\dagger:U\in\mathsf U(d)\}
	\]
	is the full real vector space $\mathsf{Herm}(\mathcal H)$.
\end{lemma}

\begin{proof}
	Let
	\begin{equation}
		V = \operatorname{span}_{\mathbb{R}}\{U X_0 U^\dagger : U \in \mathsf{U}(d)\}.   
	\end{equation}
	Since $X_0$ is not proportional to the identity, its traceless part
	$$
	X_0 - \frac{\operatorname{Tr} X_0}{d} I
	$$
	is nonzero. Moreover,
	\begin{equation}
		\int_{\mathsf U(d)}UX_0U^\dagger\,dU=\frac{\Tr X_0}{d}I, 
	\end{equation}
	Since $X_0$ is nonzero and positive, $\Tr X_0>0$; thus $I\in V$.
	Subtracting its scalar part from $X_0$ shows that
	$X_0-(\Tr X_0/d)I$ is a nonzero element of
	$V\cap\operatorname{Herm}_{\mathrm{traceless}}(\mathcal H)$.
	This intersection is invariant under unitary conjugation.  The conjugation
	representation on traceless Hermitian operators is irreducible, so the
	intersection is the entire traceless Hermitian space.  Together with
	$I\in V$, we obtain
	\begin{equation}
		V = \mathbb{R} I \oplus \operatorname{Herm}_{\text{traceless}}(\mathcal{H}) = \operatorname{Herm}(\mathcal{H}).  
	\end{equation}
\end{proof}

\begin{lemma}[Linear nonconcentration]
	\label{lem:linear-nonconcentration}
	Let $X$ be an operator component sampled under condition~(S2), and let
	$L$ be a nonzero real linear functional on the corresponding
	Hermitian-operator space.  Then
	\begin{equation}
		\Pr_{\boldsymbol{\Xi}}\{L(X)=0\}=0.
	\end{equation}
\end{lemma}
\begin{proof}
	Let $\nu$ be the unitarily invariant reference measure for $X$.  Disintegrate
	$\nu$ over the unitary orbits determined by the spectra of its operators.
	On a non-scalar orbit, choose a representative $X_0$.  By
	Lemma~\ref{lem:unitary-orbit-spanning}, the orbit of $X_0$ spans the full
	Hermitian-operator space.  Thus the real-analytic function
	\[
	U\longmapsto L(UX_0U^\dagger)
	\]
	is not identically zero.  Its zero set has Haar measure zero, and hence has
	zero invariant measure on that orbit.  The reference measure gives zero
	probability to scalar operators.  Integrating the orbitwise conclusions
	gives $\nu\{X:L(X)=0\}=0$.  Absolute continuity of the sampling law with
	respect to $\nu$ proves the assertion.
\end{proof}

\begin{corollary}
	\label{cor:sampling-support-spanning}
	Under condition~(S2), the real linear span of the support of each sampled
	operator component is the corresponding Hermitian-operator space.
\end{corollary}
\begin{proof}
	Otherwise a nonzero real linear functional would vanish on the entire
	support and hence with probability one, contradicting
	Lemma~\ref{lem:linear-nonconcentration}.
\end{proof}

\begin{lemma}[Generic nonvanishing of multilinear functions]
	\label{lem:generic-multilinear}
	Let $X_1,\ldots,X_m$ be independent operator components sampled under
	condition~(S2).  If a real multilinear function $F(X_1,\ldots,X_m)$ is not
	identically zero, then
	\begin{equation}
		\Pr_{\boldsymbol{\Xi}}\{F(X_1,\ldots,X_m)=0\}=0.
	\end{equation}
\end{lemma}
\begin{proof}
	We argue by induction on $m$.  The case $m=1$ is
	Lemma~\ref{lem:linear-nonconcentration}.  Assume the assertion for $m-1$.
	Choose a real basis $\{E_a\}$ of the first Hermitian-operator space and
	write
	\begin{equation}
		F(X_1,\ldots,X_m)=\sum_a x_{1,a}F_a(X_2,\ldots,X_m),   
	\end{equation}
	where each $F_a(X_2,\ldots,X_m):=F(E_a,X_2,...,X_m)$ is multilinear in the remaining arguments $X_2,...,X_m$.  Since $F$ is
	not identically zero, at least one $F_a$ is not identically zero.  By the
	induction hypothesis, the event on which all coefficient functions vanish
	has probability zero.  Outside this event, the map
	$X_1\mapsto F(X_1,X_2,\ldots,X_m)$ is a nonzero real linear functional.
	Lemma~\ref{lem:linear-nonconcentration} gives conditional probability zero
	for its vanishing.  Integration over the remaining variables completes
	the proof.
\end{proof}

\begin{lemma}[Nonconcentration of state differences]
	\label{lem:state-difference-nonconcentration}
	Let $R$ and $R'$ be independent state-valued components sampled under
	condition~(S2), and let $A$ be a nonzero traceless Hermitian operator. Then
	\begin{equation}
		\Pr_{\boldsymbol{\Xi}}\{\Tr[A(R-R')]=0\}=0.   
	\end{equation}
\end{lemma}
\begin{proof}
	Fix $R'$.  On a non-scalar unitary orbit with representative $R_0$, the
	function
	\[
	U\longmapsto \Tr(AUR_0U^\dagger)-\Tr(AR')
	\]
	is real analytic and is not identically zero.  Indeed, if it vanished
	identically, then $\Tr(AUR_0U^\dagger)=\Tr(AR')$ for every $U$.
	Averaging the left-hand side over the orbit gives
	$\Tr(A)\Tr(R_0)/d=0$, so $\Tr(AR')=0$.  Hence
	$\Tr(AUR_0U^\dagger)=0$ for every $U$; by
	Lemma~\ref{lem:unitary-orbit-spanning}, this would imply $A=0$, a
	contradiction.
	Thus its zero set has Haar measure zero on every non-scalar orbit.  The
	reference measure gives zero probability to scalar states, and the sampling
	law is absolutely continuous with respect to that reference measure.
	Consequently,
	\begin{equation}
		\Pr_{\boldsymbol{\Xi}}\{\Tr[A(R-R')]=0\mid R'\}=0.   
	\end{equation}
	Averaging over $R'$ proves the assertion.
\end{proof}

For the binary POVMs, write $M_{i,1}=M_i$ and $M_{i,2}=I_{I,i}-M_i$. We use the abbreviation $r_i(k_i|j_i):=P_{K_i|J_i}(k_i|j_i)$. 
By assumption, $r_i(k_i|j_i)$ are strictly positive for all $j_i,k_i$ and $i=1,2$.
The Choi operator of an
instrument element is
$C[\Gamma_{i,(j_i,k_i)}]=r_i(k_i|j_i)M_{i,j_i}^{T}\otimes\rho_{i,k_i}$.

For fixed $W_{\mathcal S}$, define
\begin{align}
	q_W(j_1,k_1,j_2,k_2)
	:=\Tr\!\left[
	W_{\mathcal S}
	\left(
	M_{1,j_1}^{T}\otimes\rho_{1,k_1}
	\otimes M_{2,j_2}^{T}\otimes\rho_{2,k_2}
	\right)
	\right].
	\label{eq:qW-definition}
\end{align}
Then the observed probability is
\begin{equation}
	p(j_1,k_1,j_2,k_2)
	=r_1(k_1|j_1)r_2(k_2|j_2)
	q_W(j_1,k_1,j_2,k_2).
	\label{eq:p-from-qW}
\end{equation}

\begin{lemma}[Detection of one-way quantum classes]
	\label{lem:one-way-quantum-detection}
	Assume condition~(S2). If $W_{\mathcal S}\notin\mathcal S_{Q,1\to2}$,
	then $Mar(\mathcal S_{Q,1\to2})$ fails with probability one. The same
	statement holds with labels $1$ and $2$ exchanged.
\end{lemma}
\begin{proof}
	Define
	\begin{align}
		\mathcal R_{O_2}(W_{\mathcal S})
		&:=\frac{1}{d_{O,2}}(\Tr_{O_2}W_{\mathcal S})\otimes I_{O_2},
		\label{eq:replacement-O2}\\
		D&:=D_{Q,1\to2}(W_{\mathcal S})
		:=W_{\mathcal S}-\mathcal R_{O_2}(W_{\mathcal S}).
		\label{eq:one-way-quantum-difference}
	\end{align}
	The definition of the one-way quantum class gives
	\begin{equation}
		W_{\mathcal S}\in\mathcal S_{Q,1\to2}
		\quad\Longleftrightarrow\quad D=0.
		\label{eq:Q12-difference-equivalence}
	\end{equation}
	Hence $D\neq0$, and
	\begin{equation}
		\Tr_{O_2}D=0.
		\label{eq:Q12-difference-partial-trace}
	\end{equation}
	
	\noindent\textbf{Step 1: A nonzero traceless coefficient on $O_2$.}\quad
	Because $D\neq0$ and Eq.~\eqref{eq:Q12-difference-partial-trace} removes the
	identity direction on $O_2$, there exist Hermitian operators
	$B_1,B_2,B_3$ on $I_1,O_1,I_2$, respectively, and a traceless Hermitian
	operator $T$ on $O_2$ such that
	\begin{equation}
		\Tr[D(B_1\otimes B_2\otimes B_3\otimes T)]\neq0.  
	\end{equation}
	Consequently,
	\begin{equation}
		H_T(X_1,R_1,X_2)
		:=\Tr[D(X_1^T\otimes R_1\otimes X_2^T\otimes T)]
		\label{eq:Q12-front-multilinear-witness}
	\end{equation}
	is a nonzero multilinear function on the full Hermitian operator space.
	Since the supports of $M_1$, $\rho_{1,1}$, $M_2$ span the corresponding Hermitian-operator spaces, $H_T$ is not identically zero on these supports.
	By
	Lemma~\ref{lem:generic-multilinear},
	\begin{equation}
		H_T(M_1,\rho_{1,1},M_2)\neq0
		\quad\text{with probability one}.
		\label{eq:Q12-front-nonzero-a.s.}
	\end{equation}
	For the sampled values of $M_1,\rho_{1,1},M_2$, define the Hermitian
	operator $A$ on $O_2$ such that
	\begin{equation}
		\Tr(AY)
		=\Tr[D(M_1^T\otimes\rho_{1,1}\otimes M_2^T\otimes Y)]
		\quad\text{for every Hermitian }Y.
		\label{eq:Q12-induced-A}
	\end{equation}
	Equation~\eqref{eq:Q12-front-nonzero-a.s.} yields $\Tr(AT)\neq0$, so $A\neq0$
	with probability one.  Moreover,
	\begin{equation}
		\Tr A
		=\Tr[D(M_1^T\otimes\rho_{1,1}\otimes M_2^T\otimes I_{O_2})]=0    
	\end{equation}
	by Eq.~\eqref{eq:Q12-difference-partial-trace}.  Thus $A$ is nonzero and
	traceless with probability one.
	
	\noindent\textbf{Step 2: Consequence of the Markov condition.}\quad
	We show that $(J_1,K_1)-J_2-K_2$ forces
	\begin{equation}
		\Tr[A(\rho_{2,1}-\rho_{2,2})]=0.
		\label{eq:Q12-Markov-difference}
	\end{equation}
	The generalized Born rule gives Eq.~\eqref{eq:p-from-qW}.  Summing over all
	of Alice's outcomes and using the trace-preserving property of her total
	instrument gives
	\begin{equation}
		P(J_2=j_2,K_2=k_2)=r_2(k_2|j_2)P(J_2=j_2).
		\label{eq:K2-given-J2}
	\end{equation}
	If $P(J_2=j_2)=0$, nonnegativity and strict positivity of $r_2$ imply that
	all joint terms with this value of $j_2$ vanish.  If $P(J_2=j_2)>0$, define
	\[
	Z(j_1,k_1,j_2):=\sum_{l=1}^2 r_2(l|j_2)q_W(j_1,k_1,j_2,l).
	\]
	If $Z=0$, both corresponding $q_W$ values vanish.  If $Z>0$, the Markov
	condition and cancellation of the strictly positive factor $r_2(k_2|j_2)$
	give $q_W(j_1,k_1,j_2,k_2)=Z$, independently of $k_2$.  In particular, for
	$j_1=j_2=k_1=1$,
	\begin{equation}
		q_W(1,1,1,1)=q_W(1,1,1,2).
		\label{eq:qW-independent-k2}
	\end{equation}
	Using Eq.~\eqref{eq:qW-definition} and the definition of $A$, this equality
	is exactly Eq.~\eqref{eq:Q12-Markov-difference}.
	
	\noindent\textbf{Step 3: Probability-zero conclusion.}\quad
	The states $\rho_{2,1}$ and $\rho_{2,2}$ are independent of
	$M_1,\rho_{1,1},M_2$ and of each other.  Conditional on the first three
	components, $A$ is nonzero and traceless with probability one.  Therefore
	Lemma~\ref{lem:state-difference-nonconcentration} gives
	\[
	\Pr_{\boldsymbol{\Xi}}\!\left\{
	\Tr[A(\rho_{2,1}-\rho_{2,2})]=0
	\,\middle|\,M_1,\rho_{1,1},M_2\right\}=0
	\]
	for almost every value of the conditioning variables.  Since the Markov
	event is contained in this zero event, integration gives
	\begin{equation}
		\Pr_{\boldsymbol{\Xi}}\{Mar(\mathcal S_{Q,1\to2})\text{ holds}\}=0.
		\label{eq:Q12-probability-zero}
	\end{equation}
	Exchanging labels $1$ and $2$ proves the reverse-direction statement.
\end{proof}

\begin{lemma}[Detection of one-way memoryless classes]
	\label{lem:one-way-memoryless-detection}
	Assume condition~(S2). If $W_{\mathcal S}\notin\mathcal S_{N,1\to2}$,
	then $Mar(\mathcal S_{N,1\to2})$ fails with probability one. The same
	statement holds with labels $1$ and $2$ exchanged.
\end{lemma}

\begin{proof}
	Assume that $W_{\mathcal S}\notin\mathcal S_{N,1\to2}$. There are two possible cases. If $W_{\mathcal S}\notin\mathcal S_{Q,1\to2}$, then Eq.~\eqref{eq:Q12-probability-zero} already gives the desired conclusion, because $Mar(\mathcal S_{N,1\to2})$ includes $Mar(\mathcal S_{Q,1\to2})$ among its conditional-independence requirements.
	
	It is therefore sufficient to prove the statement in the remaining case 
	\begin{equation} W_{\mathcal S}\in \mathcal S_{Q,1\to2}\setminus\mathcal S_{N,1\to2}. 
		\label{eq:Q-not-N-case} 
	\end{equation} 
	Because $W_{\mathcal S}\in\mathcal S_{Q,1\to2}$, the definition of the one-way quantum class gives 
	$W_{\mathcal S} = \widetilde W_{I_1O_1I_2}\otimes I_{O_2}$, 
	and
	$\widetilde W_{I_1O_1I_2} = \frac{1}{d_{O,2}} \Tr_{O_2}W_{\mathcal S}$. 
	In the following steps, we prove that $Mar(\mathcal S_{N,1\to2})$ holds with probability zero in this remaining case.

	\noindent\textbf{Step 1: Operator characterization of the one-way memoryless class.}\quad
	For each state $R$ on $\mathcal H_{O,1}$, 
	we construct an operator $E_W(R)$ on $\mathcal H_{I,1}\otimes\mathcal H_{I,2}$
	and show that it vanishes for every
	inserted state $R$ if and only if
	$W_{\mathcal S}\in\mathcal S_{N,1\to2}$.
	
	For a state $R$
	on $\mathcal H_{O,1}$, define the induced operator on $\mathcal H_{I,1}\otimes\mathcal H_{I,2}$ by
	\begin{equation}
		X_R
		:=\Tr_{O_1}\!\left[
		\widetilde W
		(I_{I_1}\otimes R\otimes I_{I_2})
		\right],
		\label{eq:XR-definition}
	\end{equation}
	and put
	$Y_R:=\Tr_{I_1}X_R$.
	For a physical process with $W_{\mathcal S}=\widetilde W\otimes I_{O_2}$,
	the no-signaling-in-time normalization, after tracing out $O_2$, reads
	$\Tr_{I_2}\widetilde W
	=d_{O,1}^{-1}(\Tr_{O_1I_2}\widetilde W)\otimes I_{O_1}$.
	Set $\rho_1:=d_{O,1}^{-1}\Tr_{O_1I_2}\widetilde W$.
	Positivity of $\widetilde W$ gives $\rho_1\geq0$, and
	$\Tr\widetilde W=d_{O,1}$ gives $\Tr\rho_1=1$.  Thus
	\begin{equation}
		\Tr_{I_2}\widetilde W=\rho_1\otimes I_{O_1}.
		\label{eq:one-way-comb-normalization}
	\end{equation}
	Moreover, $X_R$ is positive: by cyclicity under the partial trace over
	$O_1$, Eq.~\eqref{eq:XR-definition} is the partial trace of
	$(I_{I_1}\otimes R^{1/2}\otimes I_{I_2})\widetilde W
	(I_{I_1}\otimes R^{1/2}\otimes I_{I_2})$.
	Therefore, every state $R$ satisfies the relations
	\begin{align}
		\Tr_{I_2}X_R
		=\Tr_{O_1}[(\rho_1\otimes I_{O_1})(I_{I_1}\otimes R)]
		=\rho_1,\quad
		\Tr X_R=1.
		\label{eq:XR-normalization}
	\end{align}
	Define the operator
	$E_W(R):=X_R-\rho_1\otimes Y_R$.
	It is linear in $R$ and satisfies
	\begin{equation}
		\Tr_{I_1}E_W(R)=0,
		\quad
		\Tr_{I_2}E_W(R)=0.
		\label{eq:EW-definition-zero-marginals}
	\end{equation}
	
	If
	$W_{\mathcal S}\in\mathcal S_{N,1\to2}$, then
	$\widetilde W=\rho_1\otimes C[\Lambda_C]$, and
	Eq.~\eqref{eq:XR-definition} gives
	$X_R=\rho_1\otimes Y_R$ for every $R$.  Conversely, assume that $E_W(R)=0$ for every state $R$. Set
	\begin{equation}
		C:=\Tr_{I_1}\widetilde W,
		\qquad
		D:=\widetilde W-\rho_1\otimes C.
		\label{eq:memory-factorization-D}
	\end{equation}
	Then, every state $R$ satisfies the relation
	$\Tr_{O_1}[D(I_{I_1}\otimes R\otimes I_{I_2})]=0$.
	Since density operators linearly span $\mathsf{Herm}(\mathcal H_{O_1})$,
	this relation holds for every Hermitian $R$. Hilbert--Schmidt duality then gives
	$D=0$, so $\widetilde W=\rho_1\otimes C$. Positivity of $C$ follows by
	partial trace, and Eq.~\eqref{eq:one-way-comb-normalization} gives
	$\Tr_{I_2}C=I_{O_1}$. Thus $C=C[\Lambda_C]$ for a CPTP map and
	$W_{\mathcal S}\in\mathcal S_{N,1\to2}$.
	Hence
	\begin{equation}
		W_{\mathcal S}\in\mathcal S_{N,1\to2}
		\quad\Longleftrightarrow\quad
		E_W(R)=0\ \text{for every state }R.
		\label{eq:EW-definition-equivalence}
	\end{equation}
	This proves the operator characterization stated at the beginning of Step~1.
	
	\noindent\textbf{Step 2: Consequence of conditional independence.}\quad
	We show the event inclusion
	\begin{equation}
		\begin{split}
			&\{Mar(\mathcal S_{N,1\to2})\text{ holds}\}
			\\
			&\quad\subseteq
			\left\{
			\Tr\!\left[
			E_W(\rho_{1,k_1})
			(M_{1,j_1}^{T}\otimes M_{2,j_2}^{T})
			\right]=0
			\text{ for all }j_1,k_1,j_2
			\right\}.
		\end{split}
		\label{eq:OO2}
	\end{equation}
	To show this relation,
	for each value $j_1,k_1,j_2$, we prove that
	$J_1\perp J_2\mid K_1=k_1$ forces 
	$\Tr\!\left[
	E_W(\rho_{1,k_1})(M_{1,j_1}^{T}\otimes M_{2,j_2}^{T})
	\right]$
	to vanish. 
	We now
	identify the corresponding factorization difference.  For a fixed value
	$K_1=k_1$, set $R=\rho_{1,k_1}$ and define
	\begin{equation}
		h_R(j_1,j_2)
		:=\Tr\!\left[
		X_R(M_{1,j_1}^{T}\otimes M_{2,j_2}^{T})
		\right].
		\label{eq:hR-definition}
	\end{equation}
	By Eq.~\eqref{eq:XR-normalization} and POVM completeness, $h_R$ is a
	normalized joint distribution with marginals
	\begin{equation}
		h_{R,1}(j_1)=\Tr(\rho_1M_{1,j_1}^{T}),\qquad
		h_{R,2}(j_2)=\Tr(Y_RM_{2,j_2}^{T}).
		\label{eq:hR-marginals}
	\end{equation}
	
	For fixed $K_1=k_1$, define the observed conditional-probability matrix
	\begin{equation}
		A^{(k_1)}_{j_1j_2}
		:=
		P_{J_1,J_2|K_1}(j_1,j_2|k_1).   
	\end{equation}
	Let $H_R=[h_R(j_1,j_2)]_{j_1,j_2}$ denote the matrix with entries
	$h_R(j_1,j_2)$.
	Since $r_1(k_1|j_1)>0$, also $P(K_1=k_1)>0$.  Summing over
	$K_2$ in the generalized Born rule and using
	$W_{\mathcal S}=\widetilde W\otimes I_{O_2}$ gives
	$A^{(k_1)}_{j_1j_2}
	=r_1(k_1|j_1)h_R(j_1,j_2)/P(K_1=k_1)$.
	Thus $A^{(k_1)}$ is obtained from $H_R$ by positive row scaling and
	overall normalization, which preserve rank. Therefore
	$J_1\perp J_2\mid K_1=k_1$
	holds if and only if $H_R$ has rank one, or equivalently,
	\begin{equation}
		h_R(j_1,j_2)=\Tr(\rho_1M_{1,j_1}^{T})\Tr(Y_RM_{2,j_2}^{T})
		\label{eq:hR-factorization}
	\end{equation}
	for all $j_1,j_2$.
	Subtracting the right-hand side gives
	\begin{align}
		&h_R(j_1,j_2)
		-\Tr(\rho_1M_{1,j_1}^{T})\Tr(Y_RM_{2,j_2}^{T})\notag\\
		&\qquad=
		\Tr\!\left[
		E_W(R)(M_{1,j_1}^{T}\otimes M_{2,j_2}^{T})
		\right].
		\label{eq:N12-factorization-difference}
	\end{align}
	Thus, if $J_1\perp J_2\mid K_1=k_1$, then the RHS in
	Eq.~\eqref{eq:N12-factorization-difference} vanishes for every $j_1,j_2$. This
	establishes 
	\eqref{eq:OO2}.
	
	\noindent\textbf{Step 3: A nonzero multilinear witness and the probability-zero conclusion.}\quad
	Define
	\begin{equation}
		G(M_1,R,M_2)
		:=\Tr\!\left[E_W(R)(M_1^T\otimes M_2^T)\right].
		\label{eq:N12-multilinear-witness}
	\end{equation}
	This function is multilinear in $M_1,R,M_2$.  It is not identically zero.
	Indeed, if it vanished on the product of the three sampling supports, then
	Corollary~\ref{cor:sampling-support-spanning} and multilinearity would give
	\[
	\Tr[E_W(R)(A_1\otimes A_2)]=0
	\]
	for all Hermitian $R,A_1,A_2$.  Hence $E_W(R)=0$ for every $R$, contradicting
	Eqs.~\eqref{eq:Q-not-N-case} and \eqref{eq:EW-definition-equivalence}.
	Lemma~\ref{lem:generic-multilinear} therefore gives
	\begin{equation}
		\Pr_{\boldsymbol{\Xi}}\{G(M_1,\rho_{1,1},M_2)=0\}=0.
		\label{eq:N12-witness-probability-zero}
	\end{equation}
	On the other hand, the event inclusion in Eq.~\eqref{eq:OO2}, specialized
	to $j_1=j_2=k_1=1$, shows that $Mar(\mathcal S_{N,1\to2})$ forces this
	witness to vanish.  Consequently,
	\begin{equation}
		\Pr_{\boldsymbol{\Xi}}\{Mar(\mathcal S_{N,1\to2})\text{ holds}\}=0.
		\label{eq:N12-probability-zero}
	\end{equation}
	The reverse-direction statement follows by exchanging labels $1$ and $2$.
\end{proof}

\subsection{Completion of the proof}
Lemmas~\ref{lem:one-way-quantum-detection} and \ref{lem:one-way-memoryless-detection} establish Theorem~2 for the four one-way classes $\mathcal S_{Q,1\to2}$, $\mathcal S_{Q,2\to1}$, $\mathcal S_{N,1\to2}$, and $\mathcal S_{N,2\to1}$. It remains to prove the statement for $\mathcal S_Q$ and $\mathcal S_I$.

We use the intersection identities \begin{align} \mathcal S_Q &= \mathcal S_{Q,1\to2} \cap \mathcal S_{Q,2\to1}, \label{eq:SQ-intersection} \\ \mathcal S_I &= \mathcal S_{N,1\to2} \cap \mathcal S_{N,2\to1}. \label{eq:SI-intersection} \end{align}

First, suppose that $W_{\mathcal S}\notin\mathcal S_Q$. By Eq.~\eqref{eq:SQ-intersection}, there is a direction $a\to b$, where $(a,b)=(1,2)$ or $(2,1)$, such that
\begin{equation} W_{\mathcal S}\notin\mathcal S_{Q,a\to b}. \end{equation} 
The Markov condition $Mar(\mathcal S_Q)$ implies $Mar(\mathcal S_{Q,a\to b})$. Therefore, \begin{equation} \{Mar(\mathcal S_Q)\text{ holds}\} \subseteq \{Mar(\mathcal S_{Q,a\to b})\text{ holds}\}. \label{eq:SQ-event-inclusion} 
\end{equation} 
Lemma~\ref{lem:one-way-quantum-detection} shows that the event on the right-hand side of Eq.~\eqref{eq:SQ-event-inclusion} has probability zero. Consequently, 
\begin{equation} 
	\Pr_{\boldsymbol{\Xi}}\{Mar(\mathcal S_Q)\text{ holds}\}=0. \end{equation}

Next, suppose that $W_{\mathcal S}\notin\mathcal S_I$. By Eq.~\eqref{eq:SI-intersection}, there is a direction $a\to b$, where $(a,b)=(1,2)$ or $(2,1)$, such that \begin{equation} W_{\mathcal S}\notin\mathcal S_{N,a\to b}. \end{equation} The Markov condition $Mar(\mathcal S_I)$ implies $Mar(\mathcal S_{N,a\to b})$. Therefore, \begin{equation} \{Mar(\mathcal S_I)\text{ holds}\} \subseteq \{Mar(\mathcal S_{N,a\to b})\text{ holds}\}. \label{eq:SI-event-inclusion} 
\end{equation} Lemma~\ref{lem:one-way-memoryless-detection} shows that the event on the right-hand side of Eq.~\eqref{eq:SI-event-inclusion} has probability zero. Consequently, \begin{equation} \Pr_{\boldsymbol{\Xi}}\{Mar(\mathcal S_I)\text{ holds}\}=0. \end{equation} 

This completes the proof of Theorem~2.

\subsection{Rank-one projectors and pure states}
The Haar-induced probability measure on rank-one projectors is unitarily
invariant and assigns zero probability to scalar operators.  A pure state
is represented by a rank-one projector.  Therefore Haar-random rank-one
projective effects and Haar-random pure-state preparations satisfy
condition~(S2).  More generally, any sampling law absolutely continuous
with respect to these Haar-induced measures satisfies condition~(S2).
\section{Proof of Theorem 3}\label{sec:partial-proof}
\begin{lemma}[Final-output preparations]\label{lem:partial-final}
For a $1\to2$ process Bob's final preparations do not affect the induced distribution; for a parallel process neither party's output preparations affect it, provided the classical kernels remain fixed.
\end{lemma}
\begin{proof}
The respective process matrices have an identity factor on $O_2$, or on both $O_1$ and $O_2$. Each normalized output state contributes $\Tr(I\rho)=1$ to the generalized Born rule.
\end{proof}

\begin{lemma}[Detection of nonproduct parallel inputs]
\label{lem:partial-parallel-witness}
Let $W=\omega_{I_1I_2}\otimes I_{O_1}\otimes I_{O_2}$ be a
parallel process, and put
$D:=\omega_{I_1I_2}-\omega_{I_1}\otimes\omega_{I_2}$.
If $D\ne0$, then the induced distribution fails
$(J_1,K_1)\perp(J_2,K_2)$ with probability one under the sampling
of $M_1,M_2$, irrespective of the output preparation states.
\end{lemma}
\begin{proof}
By the parallel-process factorization in
Eq.~\eqref{eq:SQ-direct-factorization}, the marginal distribution of
$(J_1,J_2)$ is obtained by measuring $\omega_{I_1I_2}$; the output
preparation states contribute only their unit traces. If
$(J_1,K_1)\perp(J_2,K_2)$, then in particular $J_1\perp J_2$, and
hence
\begin{equation}
F_{\rm par}(M_1,M_2)
 :=\Tr\!\left[D(M_1^T\otimes M_2^T)\right]=0.
\label{eq:partial-parallel-witness}
\end{equation}
Because $D\ne0$ and tensor products of local Hermitian operators
span the Hermitian operators on $I_1I_2$, the displayed bilinear
function is not identically zero. The generic multilinear
nonvanishing lemma, Lemma~\ref{lem:generic-multilinear}, applied to
the independently sampled $M_1,M_2$, gives
$\Pr\{F_{\rm par}(M_1,M_2)=0\}=0$. The independence event is
contained in this zero event, proving the claim.
\end{proof}

\begin{lemma}[Detection of nonparallel one-way processes]
\label{lem:partial-nonparallel-witness}
Let $W=\widetilde W_{I_1O_1I_2}\otimes I_{O_2}$ be a physical
one-way process, and let $R_0$ be any fixed normalized state on
$O_1$. If $W\notin\mathcal S_Q$, then the induced distribution
fails $K_1-J_1-J_2-K_2$ with probability one under the sampling
of $M_1,M_2,\rho_{1,2}$ described above.
\end{lemma}
\begin{proof}
Use the existing definition of $X_R$ in
Eq.~\eqref{eq:XR-definition} and the normalization in
Eq.~\eqref{eq:XR-normalization}. For $R=\rho_{1,2}$, define
\begin{equation}
F_{\rm seq}(M_1,R,M_2)
 :=\Tr\!\left[(X_R-X_{R_0})(M_1^T\otimes M_2^T)\right].
\label{eq:partial-nonparallel-witness}
\end{equation}
This function is bilinear in the two measurement effects and affine
in the state $R$. We first verify the required implication from the
Markov condition, including zero-probability outcomes.

Put $h_R(j_1,j_2)$ as in Eq.~\eqref{eq:hR-definition}. Since the
process is one-way, summing over $k_2$ gives
\begin{equation}
P(j_1,k_1,j_2)
 =r_1(k_1\mid j_1)h_{\rho_{1,k_1}}(j_1,j_2),
\qquad
\sum_{j_2}h_{\rho_{1,k_1}}(j_1,j_2)
 =a_{j_1}:=\Tr(\rho_1M_{1,j_1}^{T}),
\label{eq:partial-one-way-joint}
\end{equation}
where the second identity is Eq.~\eqref{eq:hR-marginals} and is
independent of $k_1$. Consequently,
$P(j_1,k_1)=r_1(k_1\mid j_1)a_{j_1}$. If $a_{j_1}>0$, the
conditional independence $K_1\perp J_2\mid J_1$ implies that
$h_{\rho_{1,k_1}}(j_1,j_2)/a_{j_1}$ is independent of $k_1$:
the strictly positive factors $r_1(k_1\mid j_1)$ cancel. If
$a_{j_1}=0$, positivity of the normalized joint distributions
$h_{\rho_{1,k_1}}$ implies that all entries in that row vanish.
Thus the parallel Markov chain $K_1-J_1-J_2-K_2$ implies
\begin{equation}
h_R(j_1,j_2)=h_{R_0}(j_1,j_2)
\quad\text{for all }j_1,j_2,
\qquad\text{and in particular }F_{\rm seq}(M_1,R,M_2)=0.
\label{eq:partial-parallel-implies-zero}
\end{equation}

We next show that the zero event in
Eq.~\eqref{eq:partial-parallel-implies-zero} has probability zero.
If $X_R=X_{R_0}$ for every density operator $R$, then the linear
map $R\mapsto X_R$ vanishes on every difference of states and
therefore on all traceless Hermitian operators on $O_1$. It follows
from the definition of $X_R$ and Hilbert--Schmidt duality that
$\widetilde W$ has no traceless component on $O_1$. Hence
$\widetilde W=\omega_{I_1I_2}\otimes I_{O_1}$, with the tensor
factors placed in their original order. Physical-process
normalization makes $\omega_{I_1I_2}$ a density operator, so
$W\in\mathcal S_Q$, a contradiction. Therefore some density
operator $R_*$ satisfies $X_{R_*}-X_{R_0}\ne0$.

For this fixed $R_*$, the function
$F_{\rm seq}(M_1,R_*,M_2)$ is a nonzero bilinear function of the
two measurement effects. Lemma~\ref{lem:generic-multilinear}
therefore gives
\begin{equation}
\Pr_{M_1,M_2}\{F_{\rm seq}(M_1,R_*,M_2)=0\}=0.
\label{eq:partial-fixed-state-nonzero}
\end{equation}
For almost every measured pair $(M_1,M_2)$, define the real linear
functional $L$ on Hermitian operators on $O_1$ by
$L(R):=\Tr[X_R(M_1^T\otimes M_2^T)]$. Equation
\eqref{eq:partial-fixed-state-nonzero} implies
$L(R_*)\ne L(R_0)$. Thus the restriction of $L$ to the
traceless Hermitian subspace is nonzero: a linear functional
constant on all trace-one states could not distinguish $R_*$
from $R_0$. Write $L(R)=\Tr(AR)$, and let
$A_0:=A-(\Tr A/d_{O,1})I_{O_1}$. Then $A_0$ is nonzero and
traceless, and for every state $R$,
\begin{equation}
F_{\rm seq}(M_1,R,M_2)
 =\Tr[A_0(R-R_0)].
\label{eq:partial-affine-state-difference}
\end{equation}

For clarity, the nonconcentration statement now needed is the
fixed-reference version of
Lemma~\ref{lem:state-difference-nonconcentration}. For any fixed
state $R_0$ and any nonzero traceless $A_0$, the function
$R\mapsto\Tr[A_0(R-R_0)]$ has zero probability of vanishing when
$R$ is sampled according to the preparation component of (S2).
Indeed, on each non-scalar unitary orbit with representative $T$,
$U\mapsto\Tr[A_0UTU^\dagger]-\Tr[A_0R_0]$ is real analytic.
If it were constant zero on that orbit, averaging over $U$
would give $\Tr[A_0R_0]=0$; the unitary-orbit spanning lemma,
Lemma~\ref{lem:unitary-orbit-spanning}, would then force
$A_0=0$. Its zero set consequently has orbit measure zero.
Disintegration over the orbits and absolute continuity of the
sampling law give the asserted conditional probability zero.

Conditioning on almost every $(M_1,M_2)$ in
Eq.~\eqref{eq:partial-fixed-state-nonzero}, applying this
fixed-reference argument to the independently sampled
$R=\rho_{1,2}$, and integrating over $(M_1,M_2)$ yields
$\Pr\{F_{\rm seq}(M_1,\rho_{1,2},M_2)=0\}=0$. The Markov
event is contained in this zero event by
Eq.~\eqref{eq:partial-parallel-implies-zero}, completing the proof.
\end{proof}
\begin{proof}[Proof of Theorem~\ref{thm:partial-randomization}]
The forward Markov implications used in the proof of Theorem~1 do not require condition~(S1) and therefore apply here.
Lemma~\ref{lem:partial-final} treats fixed final preparations. Lemma~\ref{lem:partial-parallel-witness} excludes the individual condition for $W\in\mathcal S_Q\setminus\mathcal S_I$.
Since $\mathcal S_{N,1\to2}\cap\mathcal S_Q=\mathcal S_I$ (a parallel memoryless Choi factor must be $I_{O_1}\otimes\rho_2$), Lemma~\ref{lem:partial-nonparallel-witness} excludes the parallel condition for $W\in\mathcal S_{N,1\to2}\setminus\mathcal S_I$.
For $W\in\mathcal S_{Q,1\to2}\setminus(\mathcal S_Q\cup\mathcal S_{N,1\to2})$, the same lemma excludes the parallel condition.
Equations~\eqref{eq:OO2} and \eqref{eq:N12-multilinear-witness} imply that the memoryless Markov condition forces $G(M_1,\rho_{1,k_1},M_2)=0$ for every $k_1$.
Choose $k_1=2$, not $k_1=1$, because $\rho_{1,2}$ is the preparation state randomized in the present setting.
By Eq.~\eqref{eq:EW-definition-equivalence}, $G$ is a nonzero multilinear function when $W$ is not memoryless.
Applying Lemma~\ref{lem:generic-multilinear} to the independently sampled $M_1,\rho_{1,2},M_2$ gives $\Pr\{G(M_1,\rho_{1,2},M_2)=0\}=0$. The union of the two exceptional null sets is null.
Individual independence implies the parallel condition and is consequently excluded in the last two cases.
\end{proof}

\section{Information of the Settings}

In this part, we give the specific information for Alice's and Bob's seven MP settings and the parameters used to realize Charlie's strategies. Table~\ref{TE1} lists the measurement operators and preparation states used to test the Markov conditions in Eq.~(15) of the main text. 
In this case, each party makes a two-outcome POVM, 
and prepares an output state from two candidates.
Setting 1 is a predetermined reference setting.  
	In settings 2--7,
	Alice's and Bob's projective measurement bases and Alice's second preparation
	state are randomly selected, whereas Alice's first preparation state and
	Bob's two preparation states are kept fixed. The exact-distribution guarantee for independently sampled components with the specified componentwise (S2) laws is proved in Section~\ref{sec:partial-proof}. The angles listed in Table~\ref{TE1} document the selected settings, but do not by themselves determine their sampling law or establish independence.   
	Each selected setting
	is fixed before data collection and used unchanged throughout the repeated
	trials.  
The conditional probability distributions $p(K_{i}|J_{i})$
follows the condition 
\begin{equation}
	p(K_{i}=J_{i})= 0.65
\end{equation}
for $i=1,2$. That is, when Alice obtains a measurement result $J_{1}=1$, she prepares the state corresponding to $K_{1}=1$ with probability $p=0.65$, and prepares the state corresponding to $K_{1}=2$ with $1-p=0.35$. 
When Alice obtains $J_{1}=2$, she prepares the state 
$K_{1}=2$ with probability $p=0.65$, and prepares the state corresponding to $K_{1}=1$ with $1-p=0.35$. Bob's scheme is analogous. 

Thus, the experimental settings
use fewer randomized components than condition~(S2).

\setcounter{table}{0}
\renewcommand{\thetable}{E\arabic{table}}
\setcounter{figure}{0}
\renewcommand{\thefigure}{E\arabic{figure}}

\begin{table}[H]
	\centering
	\small
	\begin{tabular}{|c|c|c|c|c|}
		\hline
		\textbf{Setting} & $\mathbf{M_{1,j_1}}$ & $\mathbf{M_{2,j_2}}$ & $\mathbf{\rho_{1,k_1}}$ & $\mathbf{\rho_{2,k_2}}$ \\
		\hline
		1 & $(0,0)$ & $(45,0)$ & $(0,0)$ & $(0,0)$ \\
		& $(180,0)$ & $(-135,0)$ & $(45,0)$ & $(45,0)$ \\
		\hline
		2 & $(36,103)$ & $(15,340)$ & $(0,0)$ & $(0,0)$ \\
		& $(144,283)$ & $(165,160)$ & $(152,158)$ & $(45,0)$ \\
		\hline
		3 & $(18,76)$ & $(87,50)$ & $(0,0)$ & $(0,0)$ \\
		& $(162,258)$ & $(93,230)$ & $(149,17)$ & $(45,0)$ \\
		\hline
		4 & $(2,177)$ & $(62,143)$ & $(0,0)$ & $(0,0)$ \\
		& $(178,357)$ & $(118,323)$ & $(38,85)$ & $(45,0)$ \\
		\hline
		5 & $(12,212)$ & $(1,271)$ & $(0,0)$ & $(0,0)$ \\
		& $(168,32)$ & $(179,91)$ & $(172,236)$ & $(45,0)$ \\
		\hline
		6 & $(25,163)$ & $(8,270)$ & $(0,0)$ & $(0,0)$ \\
		& $(155,343)$ & $(172,90)$ & $(37,248)$ & $(45,0)$ \\
		\hline
		7 & $(12,292)$ & $(134,32)$ & $(0,0)$ & $(0,0)$ \\
		& $(168,112)$ & $(46,212)$ & $(75,110)$ & $(45,0)$ \\
		\hline
	\end{tabular}
	\caption{Choices of Alice's and Bob's measurement operators and preparation states. Each rank-one projector or pure state is represented by the angles of its Bloch vector in $(\theta,\phi)$ format. In each setting, the entries in the first row correspond to $J_i,K_i=1$, and those in the second row correspond to $J_i,K_i=2$. Alice's first preparation state and Bob's two preparation states are kept fixed across settings, while the other listed components are varied as described in the text.}
	\label{TE1}
\end{table}

The classical- and quantum-memory realizations described below are
	included as distinct physical realizations of the same classes in the Markov
	hierarchy.  In particular, the $\mathcal S_C$ and $\mathcal S_Q$
	realizations are both expected to be classified as $\mathcal S_Q$, whereas
	the $\mathcal S_{C,1\to2}$ and $\mathcal S_{Q,1\to2}$ realizations are both
	expected to be classified as $\mathcal S_{Q,1\to2}$.
Now we introduce the realization of Charlie's strategies. Firstly we consider the parallel ones.
For independent strategy $\mathcal{S}_{I}$, 
we let $\rho_{1}=\rho_{2}=\ket{+}\bra{+}$,
where $\ket{+}=\frac{1}{\sqrt{2}}(\ket{0}+\ket{1})$.
For $\mathcal{S}_{C}$ we let $\rho_{12}=\frac{1}{2}(\ket{0}\bra{0}\otimes\ket{0}\bra{0}+\ket{1}\bra{1}\otimes\ket{1}\bra{1})$.
For $\mathcal{S}_{Q}$ we let $\rho_{12}=\ket{\Phi}\bra{\Phi}$, where $\ket{\Phi}=\frac{1}{\sqrt{2}}(\ket{00}+\ket{11}).$ 
Note that the channel $\Lambda_{C}$ on the output system does not influence the strategy so we just ignore it (or take it as an identity channel).

For the sequential strategies, we need to determine the input state $\rho_{1}$ and the intermediate channel $\Lambda_{C}$.
For strategy $\mathcal{S}_{N,1\rightarrow 2}$, we let $\rho_{1}=\ket{0}\bra{0}$, and $\Lambda_{C}$ be an identity channel.
For strategy $\mathcal{S}_{C,1\rightarrow 2}$, we introduce binary random variable $\lambda=\{0,1\}$ between $\rho_{1}$ and $\Lambda_{C}$, with $p(\lambda=0)=\frac{1}{2}$. Let 
\begin{equation}
	\rho_{1,\lambda}=
	\begin{cases}
		\ket{0}\bra{0}, \ \lambda=0\\
		\ket{1}\bra{1}, \ \lambda=1
	\end{cases}, 
\end{equation}
and let $\Lambda_{C}$ be identity channel when $\lambda=0$ and be the bit flip channel when $\lambda=1$.
To realize the strategy $\mathcal{S}_{Q,1\rightarrow 2}$, 
we perform a quantum parallel strategy $\mathcal S_{Q}$ with probability $p=0.75$.
And with probability $1-p=0.25$, we prepare a maximally entangled state $\rho_{1M}=\ket{\Phi}\bra{\Phi}$, where subscript $M$ represents the auxiliary memory system, 
and the channel $\Lambda_{C}$ in this case is chosen to be a controlled-NOT gate, for which the control is the auxiliary system $A$ and the target is the output system $\mathcal{H}_{O,1}$ of Alice. 

Table \ref{TE3} gives the parameter of each optical device (the angles of half-wave plates and quarter-wave plates) to realize the strategies of Charlie's.

\begin{table}[H]
	\centering
	\small
	\renewcommand{\arraystretch}{1.2}
	\begin{tabular}{|c|c|c|c|c|c|c|c|c|c|}
		\hline
		\multirow{2}{*}{} & \multirow{2}{*}{\shortstack[c]{H1}} & \multirow{2}{*}{Q1} & \multirow{2}{*}{H2} & \multirow{2}{*}{Q2} & \multirow{2}{*}{H3} & \multirow{2}{*}{H4} & \multirow{2}{*}{H5} & \multirow{2}{*}{H6} & \multirow{2}{*}{H7} \\
		&  & & & & & & & & \\
		\hline
		$\mathcal S_{I}$ & 22.5 & 45 & 22.5 & 45 & 45 & 0 &0 &45 &0 \\
		\hline
		$\mathcal S_{C}$[1] & 0 & 0 & 0 & 0 & 45 & 0 & 0& 45 & 0\\ 
		\hline
		$\mathcal S_{C}$[2] & 22.5 & 45 & 22.5 & 45 & 45 & 0 &0 &45 &0 \\
		\hline
		$\mathcal S_{Q}$ & 22.5 & 45 & 0 & 0 & 0 & 0 &0 &45 &0\\
		\hline
		$\mathcal S_{N,1\to 2}$ & 0 & 0 & 22.5 & 45 & 0 & 0 &0 & 0 &0\\
		\hline
		$\mathcal S_{C,1\to 2}$[1] & 0 & 0 & 0 & 0 & 0 & 0 &0 & 0 &0\\
		\hline
		$\mathcal S_{C,1\to 2}$[2] & 0 & 0 & 22.5 & 45 & 0 & 0 &0 & 0 &0\\
		\hline
		$\mathcal S_{Q,1\to 2}$ & 22.5 & 45 & 0 & 0 & 0 & 0 & 45 & 0 &0\\
		\hline
	\end{tabular}
	\caption{Parameters to realize Charlie's different strategies on the optical platform. The Hi $(i=1,2,3,4,5,6,7)$ and Qj ($j=1,2$)  are the half-wave plates and quarter-wave plates in Fig.~3 of the main text.}
	\label{TE3}
\end{table}

\noindent\textbf{Implemented process placement.}
Now we show that our experimental setting indeed realize every target strategy.
The individual $|+\rangle\otimes|+\rangle$ input belongs to $\mathcal S_I$. The correlated classical input $(|00\rangle\langle00|+|11\rangle\langle11|)/2$ and the Bell input are nonproduct parallel states. The sequential identity channel depends on Alice's output and is memoryless but nonparallel.
In the classical-memory realization the same $\lambda$ chooses Alice's input $|\lambda\rangle$ and the channel $X^\lambda$, with both probabilities positive. Write $p=\Pr\{\lambda=0\}$, as specified in the experimental process description below, and let $Z_+$ denote the $|0\rangle\langle0|$ effect. With Alice and Bob both measuring $Z$,
\begin{equation}
	h_{|0\rangle\langle0|}(+,+)-h_{|1\rangle\langle1|}(+,+)=p\ne0,
	\qquad
	\Tr\!\left[E_W(|0\rangle\langle0|)(Z_+\otimes Z_+)\right]=p(1-p)\ne0.
	\label{eq:implemented-classical-nonzero}
\end{equation}
The first equality excludes the parallel class and the second excludes the memoryless class by Eq.~\eqref{eq:EW-definition-equivalence}. Both equalities retain the correlation between Alice's input and the conditional channel; no averaged-channel replacement is made.
In the quantum-memory realization, the Bell parallel branch has weight $q=3/4$ and the Bell-memory/CNOT branch has weight $1-q=1/4$. Taking $Z_+=|0\rangle\langle0|$ at both inputs and comparing Alice's inserted states $R_0=|0\rangle\langle0|$ and $R_1=|1\rangle\langle1|$ gives
\begin{equation}
	h_{R_0}(+,+)-h_{R_1}(+,+)=\frac{1-q}{2}=\frac18\ne0.
	\label{eq:implemented-quantum-nonparallel}
\end{equation}
For the memoryless-class test, insert $R_+=|+\rangle\langle+|$ and use $X_+=|+\rangle\langle+|$ for both input measurements. The CNOT target is invariant under the controlled bit flip. Its branch is uncorrelated in these $X$ measurements, whereas the parallel Bell branch retains $X\otimes X$ correlation. Accounting for the marginals of the mixture gives
\begin{equation}
	\Tr\!\left[E_W(R_+)(X_+\otimes X_+)\right]=\frac{q}{4}=\frac{3}{16}\ne0.
	\label{eq:implemented-quantum-nonmemoryless}
\end{equation}
Equations~\eqref{eq:implemented-quantum-nonparallel} and \eqref{eq:implemented-quantum-nonmemoryless} exclude $\mathcal S_Q$ and $\mathcal S_{N,1\to2}$, respectively. These calculations concern the specified ideal process realizations, not finite-sample data.

\section{Detailed Results of The $\chi^{2}$ Tests}\label{apdxG}
This part reports the $\chi^{2}$ values derived from the experimental
statistics for the realizations $\mathcal S_I$, $\mathcal S_C$,
$\mathcal S_Q$, $\mathcal S_{N,1\to2}$,
$\mathcal S_{C,1\to2}$, and $\mathcal S_{Q,1\to2}$.
The results were used to test the Markov conditions defining the
	hierarchy in Eq.~(15) of the main text.  
	For every implemented realization
	and every listed setting, the resulting decisions yielded the expected classification within the hierarchy.  Specifically, the $\mathcal S_C$ and $\mathcal S_Q$
	realizations were both assigned to $\mathcal S_Q$, whereas the
	$\mathcal S_{C,1\to2}$ and $\mathcal S_{Q,1\to2}$ realizations were both
	assigned to $\mathcal S_{Q,1\to2}$.  These finite-sample results are consistent with Theorem~\ref{thm:partial-randomization}; they do not establish finite-sample power or a result for arbitrary process classes.
The candidate Markov conditions are tested level by level according to
	the hierarchy in Fig.~2 of the main text.  Condition
	$J_1-K_1-J_2-K_2$ is tested through $J_1-K_1-J_2$ and
	$(J_1,K_1)-J_2-K_2$; condition $K_1-J_1-J_2-K_2$ is tested through
	$K_1-J_1-J_2$ and $(J_1,K_1)-J_2-K_2$; and condition
	$K_1-J_1-K_2-J_2$ is tested through $J_1-K_2-J_2$ and
	$K_1-J_1-(J_2,K_2)$.  If all candidate conditions at one level are rejected,
	the procedure proceeds to the next level.  The process is assigned to the
	class corresponding to the condition that is not rejected at significance
	level $\alpha=0.05$.  Non-rejection is the finite-sample decision rule used
	in the experiment and is not, by itself, proof that the exact distribution
	satisfies the corresponding Markov condition.  The critical values
	$\chi^{2}_{crit}$ are determined from the $\chi^{2}$ table.
The $\chi^{2}$ values derived from the experimental statistics are listed in Table.~\ref{chi2value}.

\begin{table}[H]
	\centering
	\small
	\renewcommand{\arraystretch}{1.2}
	\begin{tabular}{|c|c|c|c|c|c|c|c|c|}
		\hline
		\multirow{2}{*}{\shortstack[c]{Implemented\\realization}} & \multirow{2}{*}{\shortstack[c]{Markovian\\conditions}} & \multirow{2}{*}{Set.1} & \multirow{2}{*}{Set.2} & \multirow{2}{*}{Set.3} & \multirow{2}{*}{Set.4} & \multirow{2}{*}{Set.5} & \multirow{2}{*}{Set.6} & \multirow{2}{*}{Set.7} \\
		&  & & & & & & & \\
		\hline
		$\mathcal S_{I}$ & \textcircled{1} & $0.00$ & $0.09$ & $0.03$ &$0.00$ &$0.00$ &$0.00$ & $0.00$\\
		\hline
		\multirow{3}{*}{$\mathcal S_{C}$} & \textcircled{1} & 472.15 & 51.23 & 51.78 & 17.15 & 70.90 &44.83 & 53.82\\ \cline{2-9}
		& \textcircled{2} & 0.00 & 0.01 & 0.00 & 0.00 & 0.00 & 0.00 & 0.00\\ \cline{2-9}
		& \textcircled{5} & 0.00 & 0.00 & 0.00 & 0.00 & 0.00 & 0.00 & 0.00\\
		\hline
		\multirow{3}{*}{$\mathcal S_{Q}$} & \textcircled{1} & 3031.15 & 1536.00 & 42.67 & 938.71 & 3442.62 & 2996.07 & 1127.69 \\ \cline{2-9}
		& \textcircled{2} & 0.00 & 0.00 & 0.00 & 0.00 & 0.00 & 0.00 & 0.00\\ \cline{2-9}
		& \textcircled{5} & 0.00 & 0.00 & 0.00 & 0.00 & 0.00 & 0.00 & 0.00\\
		\hline
		\multirow{3}{*}{$\mathcal S_{N,1\to2}$} & \textcircled{1} & 225.49 & 3307.40 & 39.90 & 146.60 & 3497.90 & 223.16 & 371.42\\ \cline{2-9}
		& \textcircled{3} & 0.00 & 0.00 & 0.00 & 0.00 & 0.00 & 0.00 & 0.00\\ \cline{2-9}
		& \textcircled{5} & 0.00 & 0.00 & 0.00  &0.00  & 0.00 &0.00  &0.00 \\
		\hline
		\multirow{4}{*}{$\mathcal S_{C,1\to2}$} & \textcircled{2} & 55.46 & 700.06 & 22.62 & 176.12 & 1418.99 & 44.38 & 102.25\\ \cline{2-9}
		& \textcircled{3} & 121.23 & 633.89 & 48.42 & 188.03 & 1262.27 & 1164.19 & 278.38\\ \cline{2-9}
		& \textcircled{4} & 7.90 & 15.65 & 9.76 & 82.81 & 38.38 &  712.16 & 404.27\\ \cline{2-9}
		& \textcircled{5} & 0.00 & 0.00 & 0.00 & 0.00 & 0.00 & 0.00 & 0.00\\
		\hline
		\multirow{4}{*}{$\mathcal S_{Q,1\to2}$} 
		& \textcircled{2} & 55.68 & 360.85 & 8.06 & 26.47 & 905.55 &  24.04 & 23.52\\ \cline{2-9}
		&  \textcircled{3} & 3196.84& 2770.93& 25.00 & 1588.82 & 4014.75 & 5777.58  & 2087.61\\ \cline{2-9}
		& \textcircled{4} & 1919.97 & 2223.12 & 283.98& 776.05 & 3921.35 &  4429.29 & 2664.50\\ \cline{2-9}
		& \textcircled{5} & 0.00 & 0.00 & 0.00 & 0.00 & 0.00 & 0.00 & 0.00\\
		\hline
	\end{tabular}
	\caption{The $\chi^{2}$ values for each implemented process
			realization under the seven MP settings.  The Markov conditions are labeled
			as follows: \textcircled{1} $(K_1,J_1)\perp(J_2,K_2)$,
			\textcircled{2} $K_1-J_1-J_2$, \textcircled{3} $J_1-K_1-J_2$,
			\textcircled{4} $J_1-K_2-J_2$, \textcircled{5}
			$(J_1,K_1)-J_2-K_2$, and \textcircled{6}
			$K_1-J_1-(J_2,K_2)$.  The critical value is
			$\chi^{2}_{crit}=16.919$ for condition \textcircled{1},
			$\chi^{2}_{crit}=5.991$ for conditions \textcircled{2}--\textcircled{4},
			and $\chi^{2}_{crit}=12.592$ for conditions
			\textcircled{5}--\textcircled{6}.}
	\label{chi2value}
\end{table}

\end{document}